\documentclass[11pt]{article}

\usepackage{longtable, ltcaption}%
\usepackage{amsmath,amsthm,amssymb}
\usepackage{bm}
\usepackage[margin=.7in]{geometry}
\usepackage{hyperref}
\hypersetup{
    colorlinks=true,
    linkcolor=blue,
    citecolor=blue,
    urlcolor=blue
}
\usepackage{setspace,placeins,booktabs}
\usepackage{cleveref, graphicx, caption, subcaption, indentfirst,ragged2e}
\usepackage[utf8]{inputenc}
\usepackage{changes}
\setdeletedmarkup{\color{red}\sout{#1}}
\usepackage{pifont}
\newcommand{\greencheck}{\ding{51}}
\newcommand{\redx}{$\times$}   %

\usepackage{comment}
\excludecomment{additional-material}
\usepackage{xr}
\usepackage[bibstyle=numeric, citestyle=authoryear-comp, doi=false, url=false, backend=biber, maxbibnames=6, maxcitenames=4, uniquelist=false, uniquename=false, sorting=nyt, natbib=false, style=authoryear-comp, sortcites=false]{biblatex}
\renewbibmacro{in:}{}
\DeclareNameAlias{default}{family-given/given-family}
\newcommand{\citet}{\textcite}
\newcommand{\indicator}[1]{\mathbf{1}\{#1\}}
\newcommand\independent{\protect\mathpalette{\protect\independenT}{\perp}}
    \def\independenT#1#2{\mathrel{\setbox0\hbox{$#1#2$}%
    \copy0\kern-\wd0\mkern4mu\box0}}
\newcommand{\E}{\mathbb{E}}
\renewcommand{\L}{\mathrm{L}}
\renewcommand{\P}{\mathrm{P}}
\newcommand{\F}{\mathrm{F}}
\newcommand{\Q}{\mathrm{Q}}
\newcommand{\N}{\mathcal{N}}
\newcommand{\C}{\mathcal{C}^\ast}
\newcommand{\G}{\mathcal{G}^{\star}}
\newcommand{\Var}{\text{Var}}

\newcommand{\T}{T}
\newcommand{\ATT}{\text{ATT}}
\newcommand{\K}{\text{K}}

\newcommand{\oracle}[1]{\overset{\circ}{#1}} %
\renewcommand{\d}{\mathtt{d}} %

\definecolor{derekcolor}{rgb}{0.05, 0.35, 0.05} %

\newcommand{\point}[1]{}
\newcommand{\hide}[1]{}

\theoremstyle{definition}

\newtheorem{theorem}{Theorem}

\newtheorem{lemma}{Lemma}

\newtheorem{assumption}{Assumption}

\newtheorem{proposition}{Proposition}

\newtheorem{example}{Example}[section]%

\newtheorem{remark}{Remark}[section]%

\newtheorem{inneruassumption}{Assumption}
\newenvironment{namedassumption}[1]
  {\renewcommand\theinneruassumption{#1}\inneruassumption}
  {\endinneruassumption}

\crefname{figure}{Figure}{Figures}
\crefname{assumption}{Assumption}{Assumptions}
\crefname{inneruassumption}{Assumption}{Assumptions}
\crefname{appendix}{appendix}{appendices}
\Crefname{appendix}{Appendix}{Appendices}

\title{\textbf{Robust Panel Data Causal Inference by Comparison Group Selection}}

\title{\textbf{Identification-Strategy-Robust Panel Data Causal Inference}}

\title{\textbf{Beyond Parallel Trends: An Identification-Strategy-Robust Approach to Causal Inference with Panel Data}}

\author{Brantly Callaway\footnote{John Munro Godfrey, Sr.~Department of Economics, University of Georgia. \texttt{brantly.callaway@uga.edu}} \and Derek Dyal\footnote{John Munro Godfrey, Sr.~Department of Economics, University of Georgia. \texttt{ddyal@uga.edu}} \and Pedro H.C.~Sant'Anna\footnote{Department of Economics, Emory University. \texttt{pedro.santanna@emory.edu}} \and Emmanuel S.~Tsyawo\footnote{Department of Economics, Finance and Legal Studies, University of Alabama. \texttt{estsyawo@gmail.com}}}

\begin{document}

\maketitle

\abstract{In this paper, we propose a new approach to causal inference with panel data.  Instead of using panel data to adjust for differences in the distribution of unobserved heterogeneity between the treated and comparison groups, we instead use panel data to search for ``close comparison groups''---groups that are similar to the treated group in terms of pre-treatment outcomes.  Then, we compare the outcomes of the treated group to the outcomes of these close comparison groups in post-treatment periods.  We show that this approach is often \textit{identification-strategy-robust} in the sense that our approach recovers the $\ATT$ under many different non-nested panel data identification strategies, including difference-in-differences, change-in-changes, or lagged outcome unconfoundedness, among several others.  We provide related, though non-nested, results under ``time homogeneity'', where outcomes do not systematically change over time for any comparison group.  Our strategy asks more out of the research design---namely that there exist close comparison groups or time homogeneity (neither of which is required for most existing panel data approaches to causal inference)---but, when available, leads to more credible inferences.
}

\bigskip

\bigskip


\bigskip


\bigskip

\bigskip

\normalsize

\onehalfspacing
\section{Introduction}
Panel data approaches to causal inference, such as difference-in-differences, are often categorized with instrumental variables (IV) and regression discontinuity (RD) as being the main quasi-experimental approaches to causal inference.  Unlike natural experiment-based approaches to causal inference, such as IV and RD, where the full identification strategy essentially holds as a byproduct of the research design, panel data approaches to causal inference typically require auxiliary assumptions that are not implied by the research design.  A leading example is the parallel trends assumption in the context of difference-in-differences.  This distinction often results in panel data approaches being referred to as \textit{model-based} in contrast to \textit{design-based}, with the implication being that they are less credible, in the sense of relying on auxiliary assumptions not implied by the research design. For example, it is often \textit{ex ante} unclear whether the parallel trends assumption is more plausible than the assumptions required for change-in-changes or lagged outcome unconfoundedness identification strategies in any given application.

In the current paper, we retain the conventional pre/post panel data causal inference research design that aims to exploit variation in treatment timing across units and access to untreated comparison groups.  However, our aim in the current paper is to increase the credibility of causal inferences in pre/post designs by substantially relaxing the auxiliary assumptions required to go from research design to identification strategy.  We mainly consider settings where there are multiple possible comparison groups and where the researcher observes microdata on units within each group.  A leading example is the case where a researcher is interested in learning about the causal effect of a state-level policy, has access to data from multiple states, and observes underlying units such as counties, individuals, or firms.  Instead of imposing particular auxiliary assumptions to move from research design to identification strategy, we instead look for a subset of the available comparison groups that are ``very similar'' to the treated group in terms of their pre-treatment outcomes---we refer to such groups as \textit{close comparison groups}.  We show that, in many cases, close comparison groups that satisfy certain relatively simple criteria in pre-treatment periods provide valid comparison groups under many different auxiliary assumptions coming from different identification strategies.  To give a leading example, we show that, if we can find an untreated comparison group with the same distribution of untreated potential outcomes as the treated group in the period immediately before treatment, then comparing post-treatment outcomes for the treated group to the post-treatment outcomes for this group recovers the $\ATT$  under a difference-in-differences identification strategy, a change-in-changes identification strategy, or a lagged outcome unconfoundedness strategy.  Similarly, we discuss conditions on comparison groups that rationalize using them under other identification strategies, such as unit-specific linear trends, interactive fixed effects models for untreated potential outcomes, and dynamic panel models for untreated potential outcomes.\footnote{Another important motivation for panel data approaches to causal inference is that they present the opportunity for the researcher to adjust for systematic latent differences between the treated group and comparison group.  Many panel data approaches to causal inference can be thought of as trying to use panel data to adjust for these latent differences across groups.  The exact way to make these between-group adjustments is the key distinguishing feature between different panel data identification strategies and is the main source of different estimates arising from different identification strategies. Oftentimes, there is no good way for researchers to rationalize one type of adjustment over another.  Moreover, adjusting for latent differences across groups is obviously more challenging the ``more different'' the groups are.  Instead of trying to use panel data to adjust for latent differences, our approach amounts to using panel data to find groups that have the same distribution of latent characteristics.  If these groups exist, then our strategy fully obviates the need to make adjustments to make the groups more comparable.}

Our approach is conceptually much different from typical approaches to causal inference in pre/post designs, where a researcher chooses an identification strategy and uses pre-treatment periods to validate (typically via a placebo test) the identification strategy.  Instead, we propose to choose a set of comparison groups that are valid for many different auxiliary assumptions required to move from design to identification strategy.  A side-effect of choosing comparison groups that are robust to multiple identification strategies is that the number of applications where our approach can be implemented is smaller than would be the case if a researcher knew that the particular auxiliary assumptions rationalizing some identification strategy held.  For example, one criteria that we use below is that the mean outcome for a candidate comparison group is the same as the mean outcome for the treated group in the pre-treatment period.  We show that, if parallel trends holds, then a level-comparison of mean outcomes between this comparison group and the treated group in the post-treatment period is equal to the $\ATT$.  However, it is possible that there are no comparison groups that satisfy this criteria---a DiD identification strategy would still work in this case, but our approach would not deliver any estimate of the $\ATT$.
 We view this as a feature, rather than a bug, of our approach---if such a comparison group exists, results using our approach will be more robust to different identification strategies (for example, our approach would be robust to certain violations of parallel trends),
 and effectively less dependent on functional form assumptions.  This suggests that inferences using our approach, when it is feasible, are likely to be substantially more credible in a pre/post research design than inferences that rely on auxiliary assumptions such as parallel trends.

We provide several extensions of these results.  First, we provide related results for settings with time homogeneity ---where the outcomes do not systematically change over time for any of the comparison groups.  We show that time homogeneity often provides robustness to multiple identification strategies as well, though the exact kind of robustness differs from using close comparison groups.  Thus, we also propose a way to combine both close comparison groups and time homogeneity to further increase robustness.  Second, we discuss how to extend our results to settings where close comparison groups are formed conditional on covariates.  Third, we discuss how to extend our results to settings with more than one treated group and variation in treatment timing across units.

\subsubsection*{Related Work}

Our approach is related to synthetic control-type methods (\citet{abadie-diamond-hainmueller-2010,arkhangelsky-etal-2021}), which construct weights on untreated units in order to match the treated group's pre-treatment characteristics.  A key difference between our approach and the synthetic control method is that we consider a setting with multiple possible comparison groups where all units within the same comparison group ``count equally''.  This explains why our results apply in settings with few time periods, while synthetic control methods typically require many time periods to estimate unit-specific weights.  The other key difference is that we focus on finding a set of close comparison groups, any of which could be used as a valid comparison group.  This is a more stringent requirement on the research design than finding a synthetic control---i.e., there may be applications where the set of close comparison groups is empty, but one could construct a synthetic control.  However, when close comparison groups exist, our approach leads to more robust inferences as we show that these groups can be used under many different identification strategies.  %
\citet{kellogg-mogstad-pouliot-torgovitsky-2021} discuss the trade-off between interpolation bias from synthetic control and extrapolation bias from matching.   %
\citet{gunsilius-2023,gunsilius-2025} use micro-data and construct group-level comparison groups based on applying weights to either the quantiles or cdfs; this is related to our idea of forming close comparison groups, but we employ a tighter definition of comparison groups.%

Our paper is more closely related to recent work that has explored using multiple comparison groups in panel data settings.  Part of this literature has used multiple comparison groups to increase efficiency (\citet{marcus-santanna-2021,chen-santanna-xie-2025}).  Another part of this literature has used multiple comparison groups to relax the parallel trends assumption (\citet{callaway-tsyawo-2023,sun-xie-zhang-2025,rincon-song-2025,liu-2025}), often with an interactive fixed effects model for untreated potential outcomes as motivation.

Our paper is related to a growing literature on robust causal inference with panel data (\citet{manski-pepper-2018,rambachan-roth-2023,arkhangelsky-imbens-2022,arkhangelsky-imbens-lei-luo-2024,leavitt-hatfield-2025,athey-imbens-qu-viviano-2025}). Finally, our results on time homogeneity are related to the idea of using ``time as an instrument'' (\citet{chernozhukov-val-hahn-newey-2013}) and regression discontinuity in time (\citet{hausman-rapson-2018,cattaneo-diaz-titiunik-2024}).
\section{Setup}

\subsection{Notation and Assumptions}

We consider a setting where a researcher observes $\T$ periods of panel data.  %
Let $D_{it}$ denote unit $i$'s treatment status in time period $t$, and let $D_i = (D_{i1},D_{i2},\ldots, D_{i\T})'$ denote unit $i$'s vector of treatments across all periods.  Let $G_i$ denote a unit's group, and let $\mathcal{G}$ denote the support of $G_i$.
We consider the case where all units in the same group experience the same treatment regime.  This means that, for $g \in \mathcal{G}$, we can write $\mathbf{d}(g)$ to indicate the treatment regime experienced by group $g$.  Next, let $Y_{it}$ denote unit $i$'s observed outcome in period $t$.  Next, we define potential outcomes.  Let $Y_{it}(\mathbf{d})$ denote the outcome unit $i$ would experience in time period $t$ if it experienced treatment regime $\mathbf{d} = (d_1,d_2,\ldots,d_\T)'$.  We use the notation $\mathbf{0} = (0, 0, \ldots, 0)'$ to denote the treatment regime of being untreated in all periods, and refer to $Y_{it}(\mathbf{0})$ as unit $i$'s untreated potential outcome in time period $t$.  And, slightly abusing notation, we often use the shorthand notation $Y_{it}(g)$ for $Y_{it}(\mathbf{d}(g))$, which is the potential outcome for unit $i$ in time period $t$ if it experienced the treatment regime implied by being in group $g$.  We maintain the following assumptions throughout the paper:

\begin{assumption}[Sampling] \label{ass:sampling} The observed data consists of $\big\{Y_{i1}, Y_{i2},\ldots,Y_{i\T}, D_{i1}, D_{i2}, \ldots, D_{i\T}, G_i\big\}_{i=1}^n$, where $n$ denotes the sample size, which are independent and identically distributed.  The outcomes are continuously distributed.
\end{assumption}

\begin{assumption}[Treatment and Groups] \label{ass:groups} All units in the same group experience the same treatment regime.  That is, $G_i$ fully determines $(D_{i1},D_{i2},\ldots,D_{i\T})$.  No units are treated in the first period, i.e., $D_{i1}=0$ for all units.
\end{assumption}

\begin{assumption}[No Anticipation] \label{ass:no-anticipation} For all $t \in \{1,\ldots,\T\}$ and for any unit $i$, $Y_{it}(\mathbf{d}) = Y_{it}(\mathbf{d}^t)$ where $\mathbf{d}^t = (d_1, \ldots, d_t)$
\end{assumption}
\Cref{ass:sampling} says that we observe $\T$ periods of panel data.  \Cref{ass:groups} connects groups and treatments.  \Cref{ass:no-anticipation} says that potential outcomes in a particular period $t$ can depend on the treatment status in the current period or previous periods, but they do not depend on future treatments.

Next, we provide an additional assumption that focuses the discussion in the main text, though we discuss how to relax it later in the paper.

\begin{assumption}[Single Treated Group] \label{ass:single-treated} Only one group participates in the treatment.  Without loss of generality, we use $G=1$ to denote this group.
\end{assumption}
\Cref{ass:single-treated} says that there is a single group that receives the treatment.  Let $t^*_{g}$ denote the time period that group $g$ first participates in the treatment.  Let $\bar{\mathcal{G}}$ denote the set of groups that do not participate in the treatment in any period.  Under \Cref{ass:single-treated}, for simplicity, we let $Y_{it}(1)$ denote the potential corresponding to the treatment regime of being untreated up through period $t^*_1-1$, becoming treated in period $t^*_1$, and then following the same treatment regime experienced by units in group 1.  We do not use the assumption of staggered treatment adoption, where treatment timing can vary across units, but once a unit becomes treated, it remains treated in subsequent periods.  Staggered treatment adoption, however, is an important special case and has been widely considered in the literature on difference-in-differences (\citet{goodman-bacon-2021,callaway-santanna-2021}).  Oftentimes in that literature, the groups are based on treatment timing.
It is worth pointing out that our notion of groups can be more general than forming groups solely on the basis of treatment timing, as our setting allows for groups that could experience the same treatment regime (e.g., if we define groups on the basis of state, but several states become treated at the same point in time).

\subsubsection{Sources of Multiple Comparison Groups}

The approach that we consider relies on there being multiple possible comparison groups and that we observe underlying units with each group.  Next, we provide some leading examples of where these groups could come from.

\begin{example}[State-level policies] Consider the case where we are interested in the effect of a policy implemented by a particular state.  In this case, $\mathcal{G}$ could be the set of all states, and $\bar{\mathcal{G}}$ could include all other states that did not implement the policy.  We also suppose that we observe underlying units in each state, e.g., counties, individuals, or firms.
\end{example}
\begin{example}[Staggered treatment adoption] Suppose that different units become treated at different points in time.  Let $G_i$ denote the time period when unit $i$ becomes treated.  In this case, $\mathcal{G}$ could include all timing-groups and $\bar{\mathcal{G}}$ could include all timing-groups that are not-yet-treated by period $t^*_1$.
\end{example}

\subsubsection{Target Parameters}

Following the majority of the literature on panel data causal inference, we target the average treatment effect on the treated ($\ATT$).  The DiD literature has often focused on ``group-time average treatment effects''---$\ATT$'s that can vary by group and time period.  In our setting, when there is a single treated group, we can simply index $\ATT$'s by time:
\begin{align*}
    \ATT_t = \E[Y_t(1) - Y_t(0) \mid G=1]
\end{align*}
which is the average treatment effect in period $t$ for group 1.  For the remainder of this section, we mainly focus on $\ATT_{t^*_1}$---the average treatment effect for the treated group in the period it becomes treated---though it is straightforward to extend the arguments presented in this section to recover $\ATT_t$ in other post-treatment periods.  We discuss settings with more groups and more post-treatment periods in \Cref{sec:more-extensions}.  %

\subsection{Overview of Identification Strategy}

The central idea of our identification strategy is to choose a subset of the possible comparison groups based on being similar to the treated group in pre-treatment periods.  This is much different from the way that panel data is typically used in modern approaches to causal inference.  Instead of using panel data to somehow account for differences in unobserved heterogeneity between different groups, we use panel data to find very similar comparison groups.  Then, under certain conditions, groups being ``similar'' in terms of pre-treatment outcomes can imply that they are also similar in terms of unobserved heterogeneity, further implying that they should be alike in post-treatment periods even if the mapping from unobservables to untreated potential outcomes changes over time.  More generally, the upshot of having comparable groups is that our approach is robust to many different panel data causal inference identification strategies.

For our main results in this section, we consider the subset of comparison groups defined as follows:
\begin{align}
    \mathcal{G}^* = \big\{ g \in \bar{\mathcal{G}} : (Y_{t^*_1-1} \mid G=1) \overset{d}{=} (Y_{t^*_1-1} \mid G=g) \big\} \label{eqn:Gstar-def}
\end{align}
which is the set of groups that have the same distribution of outcomes in the pre-treatment period $t^*_1-1$ as the treated group has.  We refer to groups that satisfy the condition in \Cref{eqn:Gstar-def} as \textit{close comparison groups}.  For any group $g \in \bar{\mathcal{G}}$, define
\begin{align*}
    \tau_g = \E[Y_{t^*_1} \mid G=1] - \E[Y_{t^*_1} \mid G=g]
\end{align*}
i.e., $\tau_g$ is the difference in mean outcomes between group 1 and group $g$. Our main result is that, for any group $g \in \mathcal{G}^*$,
\begin{align*}
    \tau_g = \ATT_{t^*_1}
\end{align*}
if any of several leading panel data identification strategies hold.  In other words, if we compare the average outcome in levels for the treated group to the average outcome in levels for group $g \in \mathcal{G}^*$ (i.e., to a close comparison group), then that comparison delivers the average treatment effect for the treated group.  Because there may be multiple close comparison groups, the ultimate estimand that we use is
\begin{align} \tau = \E[Y_{t^*_1} \mid G = 1] - \displaystyle \sum_{g \in \mathcal{G}^*} w(g) \cdot \E[Y_{t^*_1} \mid G=g] \label{eqn:estimand}
\end{align}
where $w(g)$ are weights, which can depend on things like the relative frequencies of each close comparison group.\footnote{One example of $w(g)$ is $w(g) = \P(G=g \mid G \in \mathcal{G}^*)$---in this case, the right hand side of \Cref{eqn:estimand} simplifies to $\E[Y_{t^*_1} \mid G \in \mathcal{G}^*]$, which is simply the average untreated potential outcome across all units in that are in any comparison group in $\mathcal{G}^*$.   Most of our identification results below go through for any weighting function $w(g)$, so choosing the weights is often more of an estimation problem, with good choices of weights leading to more precise estimates of $ATT_{t^*_1}$.}

\subsection{Five Leading Panel Data Causal Inference Identification Strategies}

In this subsection, we briefly review five leading identification strategies for panel data causal inference.  Later, we will connect these identification strategies to our approach of using close comparison groups.

\bigskip

\noindent \textit{\bfseries Identification Strategy 1: Difference-in-Differences}
\bigskip

The key assumption underlying a difference-in-differences identification strategy is the following parallel trends assumption:
\begin{namedassumption}{DiD-1} \label{ass:pt}
For all $g \in \mathcal{G}$,
\begin{align*}
    \E[\Delta Y_{t^*_1}(0) \mid G=g] = \E[\Delta Y_{t^*_1}(0)]
\end{align*}
\end{namedassumption}
\Cref{ass:pt} says that the trend in untreated potential outcomes over time is the same for all groups.  It follows immediately that $\ATT$ is identified under the parallel trends assumption.  In particular, for any $g \in \bar{\mathcal{G}}$,
\begin{align*}
    \ATT_{t^*_1} &= \E[\Delta Y_{t^*_1} \mid G=1] - \E[\Delta Y_{t^*_1} \mid G=g]
\end{align*}

\bigskip

\noindent \textit{\bfseries Identification Strategy 2: Change-in-Changes}
\bigskip

The main assumption for a change-in-changes identification strategy (\citet{athey-imbens-2006}) is:
\begin{namedassumption}{CiC-1} \label{ass:cic}
For any $y \in \text{support}(Y_{t^*_1-1}|G=1)$
\begin{align*}
    Q_{Y_{t^*_1}(0)|G=1}\big(\F_{Y_{t^*_1-1}(0)|G=1}(y)\big) = Q_{Y_{t^*_1}(0)|G=g}\big(\F_{Y_{t^*_1-1}(0)|G=g}(y)\big) \text{ for all $g \in \mathcal{G}$};
\end{align*}
\end{namedassumption}
See, in particular, Theorem 3.1 of \citet{athey-imbens-2006}.\footnote{\citet{athey-imbens-2006} derive this condition from an underlying model where untreated potential outcomes are generated by a non-separable but invertible function of scalar unobservables whose distribution can be different across groups but does not change across time, and under a support assumption.}  \Cref{ass:cic} can be thought of as a nonlinear way to account for trends in untreated potential outcomes over time---note that, like DiD, the left hand side is related to how untreated potential outcomes change over time for the treated group but is not identified by the sampling process, while the right hand side involves the observed change in untreated potential outcomes across time periods for an untreated comparison group.  CiC identification strategies also require the following support condition:
\begin{namedassumption}{CiC-2}
\begin{align*}
    \text{support}(Y_{t^*_1-1} \mid G=1) \subseteq \text{support}(Y_{t^*_1-1} \mid G=g)
\end{align*}
Under these assumptions, one can show that the $\ATT$ is identified and given by
\begin{align*}
    \ATT_{t^*_1} = \E[Y_{t^*_1} \mid G=1] - \E\Big[ Q_{Y_{t^*_1}|G=g}\big(\F_{Y_{t^*_1-1}|G=g}(Y_{t^*_1-1})\big) \Bigm| G=1 \Big]
\end{align*}
\end{namedassumption}

\bigskip

\noindent \textit{\bfseries Identification Strategy 3: Lagged Outcome Unconfoundedness}
\bigskip

The main idea of a lagged outcome unconfoundedness identification strategy (see, e.g., \citet{ding-li-2019,callaway-li-2023,powell-griffin-wolfson-2023}) is to compare post-treatment outcomes of treated and untreated units that have the same pre-treatment outcome.  The main identifying assumption is:
\begin{namedassumption}{LOU-1} \label{ass:lou}
    $Y_{t^*_1}(0) \independent G \mid Y_{t^*_1-1}(0)$
\end{namedassumption}
This is an unconfoundedness assumption, but where the conditioning variables are pre-treatment outcomes.  For recovering $\ATT_{t^*_1}$, an important implication of this assumption is that
\begin{align} \label{eqn:lou}
    \E[Y_{t^*_1}(0) \mid Y_{t^*_1-1}(0), G=1] = \E[Y_{t^*_1}(0) \mid Y_{t^*_1-1}(0),G=g] \ \text{for all $g \in \bar{\mathcal{G}}$}
\end{align}
The other important condition for lagged outcome unconfoundedness is the following overlap condition
\begin{namedassumption}{LOU-2} For all $g \in \bar{\mathcal{G}}$ and $y \in \text{support}(Y_{t^*_1-1})$, $\P(G=1|Y_{t^*_1-1}(0)=y, G\in\{1,g\}) < 1$
\end{namedassumption}
This assumption says that, for all treated units, there exist units in group $g \in \bar{\mathcal{G}}$ that have the same pre-treatment outcomes.  Under these assumptions, $\ATT_{t^*_1}$ is identified and given by
\begin{align*}
    \ATT_{t^*_1} = \E[Y_{t^*_1} \mid G=1] - \E\big[\E[Y_{t^*_1} \mid Y_{t^*_1-1},G=g] \bigm| G=1\big]
\end{align*}

\bigskip
\noindent \textit{\bfseries Identification Strategy 4: Interactive Fixed Effects}
\bigskip

Next, we consider an interactive fixed effects model for untreated potential outcomes based on the following assumption:
\begin{namedassumption}{IFE-1}\label{ass:ife-model} For $t=1,\ldots,T$,
\begin{align}
   Y_{it}(0) = \theta_t + \eta_i + \lambda_i'F_t + e_{it} \label{eqn:ife}
\end{align}
where $\lambda_i$ is an $R$-dimensional vector of time-invariant heterogeneity and $F_t$ is an $R$-dimensional vector of fixed coefficients on the $\lambda_i$---$F_t$ is often referred to as a factor and $\lambda_i$ is often referred to as a factor loading. $e_{it}$ satisfies the condition that $\E[e_t \mid G=g] = 0$ for all $g \in \mathcal{G}$.
\end{namedassumption}
This condition says that the groups can be systematically different from each other in terms of $\eta_i$ and $\lambda_i$, but not in terms of the time-varying error term $e_{it}$.  Interactive fixed effects models for untreated potential outcomes have been popular in papers that try to relax the parallel trends assumption (e.g., \citet{abadie-diamond-hainmueller-2010,gobillon-magnac-2016,xu-2017,callaway-tsyawo-2023}, among others).  In general, in settings with only two time periods, $\ATT_{t^*_1}$ is not identified if untreated potential outcomes are generated by the interactive fixed effects model in \Cref{eqn:ife}.

For the discussion that follows, it is often convenient to stack untreated potential outcomes across time periods and take expectations conditional on group, i.e., let $\mathbf{Y}(0) = (Y_1(0), Y_2(0), \ldots, Y_\T(0))'$, and then consider
\begin{align} \label{eqn:ife-mean}
    \E[\mathbf{Y}(0) \mid G=g] = \bm{\theta} + \mathbf{F} \mathbf{\Lambda}_g
\end{align}
where
\begin{align*}
    \E[\mathbf{Y}(0) \mid G=g] := \begin{pmatrix}
        \E[Y_1(0) \mid G=g] \\
        \vdots \\
        \E[Y_\T(0) \mid G=g]
    \end{pmatrix},
    \bm{\theta} := \begin{pmatrix} \theta_1 \\ \vdots \\ \theta_\T \end{pmatrix},
    \mathbf{F} := \begin{pmatrix} 1 & F_1' \\ \vdots & \vdots \\ 1 & F_\T' \end{pmatrix},
    \mathbf{\Lambda}_g := \begin{pmatrix} \E[\eta \mid G=g] \\ \E[\lambda \mid G=g] \end{pmatrix}
\end{align*}
where $\mathbf{\Lambda}_g$ is an $R+1$ dimensional vector.  $\mathbf{F}$ is a $\T \times (R+1)$ matrix.  We can also stack these across groups to get
\begin{align*}
    \begin{bmatrix} \E[\mathbf{Y}(0) \mid G=g]' \end{bmatrix}_{g \in \mathcal{G}} = \mathbf{1}_{|\mathcal{G}|} \bm{\theta}' + \mathbf{\Lambda} \mathbf{F}' \text{ where } \mathbf{\Lambda} := \begin{bmatrix} \mathbf{\Lambda}_g' \end{bmatrix}_{g \in \mathcal{G}}
\end{align*}

Next, define $\mathbf{Y}^{\text{pre}(t)}(0) = (Y_1(0), \ldots, Y_{t-1}(0))'$ to be the vector of untreated potential outcomes in periods before $t$, and consider
\begin{align} \label{eqn:ife-mean-pre}
    \E[\mathbf{Y}^{\text{pre}(t)}(0) \mid G=g] = \bm{\theta}^{\text{pre}(t)} + \mathbf{F}^{\text{pre}(t)} \mathbf{\Lambda}_g
\end{align}
where
\begin{align*}
    \bm{\theta}^{\text{pre}(t^*_1)} := \begin{pmatrix} \theta_1 \\ \vdots \\ \theta_{t^*_1-1} \end{pmatrix}, \quad
    \mathbf{F}^{\text{pre}(t)} := \begin{pmatrix} 1 & F_1' \\ \vdots & \vdots \\ 1 & F_{t^*_1-1}' \end{pmatrix}
\end{align*}
For our results on close comparison groups, we consider the following assumption.
\begin{namedassumption}{IFE-2*}\label{ass:ife-rank} $\text{rank}\big(\mathbf{F}^{\text{pre}(t^*_1)}\big) = \text{rank}\big(\mathbf{F}\big)$.
\end{namedassumption}
We add the $*$ superscript to indicate that we do not maintain this assumption every time that we discuss the IFE identification strategy, but rather only invoke it in certain contexts.  \Cref{ass:ife-rank} effectively requires that there be the same amount of variation in the factors in pre-treatment periods relative to the variation in the factors in post-treatment periods.  A necessary condition for \Cref{ass:ife-rank} is that the number of pre-treatment periods be at least as large as the number of unobserved heterogeneity terms, i.e., $t^*_1 - 1 \geq R+1$.  Most importantly, it rules out some factor ``turning on'' in post-treatment periods that did not affect outcomes in pre-treatment periods.

Additionally, notice that we can view \Cref{eqn:ife-mean,eqn:ife-mean-pre} as maps $\begin{pmatrix} \E[\eta \mid G=g] \\ \E[\lambda \mid G=g] \end{pmatrix} \mapsto \E[\mathbf{Y}(0) \mid G=g]$ and $\begin{pmatrix} \E[\eta \mid G=g] \\ \E[\lambda \mid G=g] \end{pmatrix} \mapsto \E[\mathbf{Y}^{\text{pre}(t)}(0) \mid G=g]$, respectively.  The rank condition in \Cref{ass:ife-rank} implies that both maps are injective, and, hence, that groups that have the same mean of pre-treatment outcomes also have the same mean of $\eta$ and $\lambda$.

Let $\bar{\mathbf{\Lambda}} = \begin{bmatrix} \mathbf{\Lambda}_g' \end{bmatrix}_{g \in \bar{\mathcal{G}}}$.  For our results on time homogeneity, we consider the following assumption.
\begin{namedassumption}{IFE-3*}\label{ass:ife-comparison-group-rank} $\text{rank}(\bar{\mathbf{\Lambda}}) = R+1$.
\end{namedassumption}
Like \Cref{ass:ife-rank}, we do not maintain \Cref{ass:ife-comparison-group-rank} throughout the paper, but only invoke it in certain contexts.  \Cref{ass:ife-comparison-group-rank} requires that there be enough independent variation in unobserved heterogeneity among the comparison groups relative to the number of interactive fixed effects.

\bigskip

\noindent \textit{\bfseries Identification Strategy 5: Latent Unconfoundedness}
\bigskip

An even more general identification strategy is to take the assumption of latent unconfoundedness as a starting point:
\begin{align} \label{eqn:latent-unconfoundedness}
    Y_{t}(0) \independent G \mid \xi
\end{align}
where $\xi$ is an $L$-dimensional vector of unobserved heterogeneity.  Latent unconfoundedness can equivalently be stated as in the following assumption
\begin{namedassumption}{Lat.~Unc.-1}\label{ass:latent-unconfoundedness} The following model for untreated potential outcomes holds
\begin{align*}
    Y_{it}(0) = h_t(\xi_i) + e_{it}
\end{align*}
where $h_t$ is left unrestricted, the distribution of $\xi$ can vary across groups, but $e_t \independent G \mid \xi$.
\end{namedassumption}
The motivation for \Cref{ass:latent-unconfoundedness} is similar to that of the interactive fixed effects model discussed above, except that it does not impose the functional form assumption of the IFE model.  Like for the case of IFE, in general, $\text{ATT}_{t^*_1}$ is not identified under latent unconfoundedness with only two time periods.  Let
\begin{align*}
    \mathbf{H}(\xi) := \begin{pmatrix} h_1(\xi), \ldots, h_\T(\xi) \end{pmatrix}, \qquad \mathbf{H}^{\text{pre}(t)}(\xi) := \begin{pmatrix} h_1(\xi), \ldots, h_{t-1}(\xi) \end{pmatrix}
\end{align*}
which are $\T$-dimensional and $t-1$-dimensional, respectively.  Without loss of generality, we define $\xi$ so that $\mathbf{H}(\xi)$ is injective---this effectively controls the dimension of $\xi$ (where the maximum possible dimension of $\xi$ is $\T$, but it can be smaller), similar to our discussion above about $\text{rank}\big(\mathbf{F}\big)$ determining the number of factors.

\begin{namedassumption}{Lat.~Unc.-2} \label{ass:latent-unconfoundedness-injectivity} $\mathbf{H}^{\text{pre}(t^*_1)}$ is injective.
\end{namedassumption}

\Cref{ass:latent-unconfoundedness-injectivity} play a similar role as \Cref{ass:ife-rank} discussed above for IFE.  It mainly rules out components of $\xi$ ``turning on'' in post-treatment periods.  The main implication of \Cref{ass:latent-unconfoundedness-injectivity} that we use below is that it implies that, for $t \geq t^*_1$, $h_t(\xi) = \phi_t\big(\mathbf{H}^{\text{pre}(t^*_1)}(\xi)\big)$ for some function $\phi_t$.  Note that $\mathbf{H}^{\text{pre}(t^*_1)}$ cannot be injective if $\text{dim}(\xi) > t^*_1-1$.

\section{Robust Identification with Close Comparison Groups} \label{sec:close-comparison-groups}

This section contains our main results on the robustness of using close comparison groups.  We start with the case where close comparison groups are based on having the same distribution of outcomes in the period immediately preceding the treatment.  Then, we discuss extensions where close comparison groups are (i) based on having only the same mean outcomes in a single pre-treatment period, (ii) based on having the same mean outcomes over several pre-treatment periods, and (iii) based on having the same distribution of outcomes over several pre-treatment periods.

\begin{theorem} \label{thm:1} Under \Cref{ass:sampling,ass:groups,ass:no-anticipation,ass:single-treated}, if $\mathcal{G}^*$ is non-empty, and if any of the following identification strategies hold: (i) difference-in-differences, (ii) lagged outcome unconfoundedness, or (iii) change-in-changes, then
    \begin{align*}
        \tau = \ATT_{t^*_1}
    \end{align*}
    In contrast, under (iv) IFE, or (v) latent unconfoundedness, it is possible that $\tau \neq \ATT_{t^*_1}$.
\end{theorem}

The proof of \Cref{thm:1} is provided in Appendix \ref{app:proofs}.  The first part of \Cref{thm:1} establishes the central result of our paper.  It shows that, if we are able to find any comparison groups that satisfy the condition of having the same distribution of observed outcomes in the pre-treatment period, then we can simply compare the average outcome for the treated group to the average outcome for these groups to recover the $\ATT$ in a way that is robust across several different leading identification strategies.  In other words, in applications where there are ``close comparison groups'', then it is possible to recover $\ATT$ without having to take a stand on the ``correct'' identification strategy, at least among DiD, CiC, and LOU---three leading identification strategies that are used in this context.  This sort of strategy may not be feasible in all applications---because there may not be any available close comparison groups---but, in applications where close comparison groups are available, this result shows that recovering the $\ATT$ is substantially more robust than in applications that rely more strongly on auxiliary assumptions like parallel trends for identification.

The identification strategy discussed here is also transparent and simple to explain to non-experts.  It suggests using pre-treatment data to find comparison groups that are very similar to the treated group, and then to compare the level of the outcomes for the treated group to the level of the outcome for very similar comparison groups.  For example, in an application where a state implements a policy and where we have microdata from every state, this approach suggests using the pre-treatment data to find other states that have similar outcome distributions and to discard states that have different distributions of outcomes in the pre-treatment period.  The simple intuition of our approach contrasts with most panel data approaches to causal inference.  For example, at least arguably, comparing trends in outcomes over time for the treated group relative to untreated comparison states (as would be done with DiD) is unnatural outside of the context of linear fixed effects models for untreated potential outcomes and can be hard to explain or rationalize to non-economists.  These strategies are also conceptually different in that they effectively use the pre-treatment data to ``adjust for'' differences in the levels of outcomes across groups, even if the distributions of these outcomes are quite different across groups.  Adjusting for differences between groups is a much harder task---and, as emphasized above, often requires making auxiliary assumptions that can be hard to rationalize---than comparing similar groups.  Our approach is simply to check if we have access to similar groups and, if we do, to compare their outcomes in post-treatment periods.

The second part of \Cref{thm:1} shows that forming $\mathcal{G}^*$ and comparing post-treatment outcomes for the treatment group relative to outcomes for groups in $\mathcal{G}^*$ is not a panacea---under IFE or latent unconfoundedness, this comparison is not guaranteed to deliver the $\ATT$.  The intuition for both of these results is that groups having the same distribution of outcomes in the pre-treatment period is not strong enough to guarantee that these groups have the same distribution of time-invariant unobservables, particularly in settings where the dimension of the time-invariant unobservables is greater than one.  To foreshadow some of our results in the next section, it is interesting to notice that, as discussed above, in general, the $\ATT$ is not identified for either of these identification strategies when there are only two time periods.  We will show below that, if there are more pre-treatment periods, at least under certain conditions, one can use an alternative definition of close comparison groups (involving groups have similar outcomes across multiple pre-treatment periods) that can rationalize the $\ATT$ as being equal to the difference between post-treatment outcomes for the treated group relative to post-treatment outcomes for groups in this alternative version of close comparison groups.

\begin{remark}[Analogy to de-biasing using pre-treatment period] \label{rem:debiasing-pre}
    Difference-in-differences, change-in-changes, and lagged outcome unconfoundedness identification strategies all have representations like:
    \begin{align*}
        \ATT_{t^*_1} = \E[Y_{t^*_1} \mid G=1] - \E[Y_{t^*_1} \mid G=g] - \text{de-biasing term from period $t^*_1-1$ differences}
    \end{align*}
    where, for example, the de-biasing term for difference-in-differences is $\E[Y_{t^*_1-1} \mid G=1] - \E[Y_{t^*_1-1} \mid G=g]$.  The particular de-biasing terms for change-in-changes and lagged outcome unconfoundedness are different.  Instead of de-biasing, the approach that we proposed above amounts to looking for groups where we do not need to do any de-biasing using the pre-treatment period, and our criteria for being a close comparison group implied that the de-biasing term for all three of these identification strategies is zero.
\end{remark}

\begin{remark}[Post-treatment placebo tests]
    When $|\mathcal{G}^*| > 1$, i.e., when there is more than one close comparison group, there are post-treatment testable implications of our approach.  In particular, if any of the difference-in-differences, change-in-changes, or lagged outcome unconfoundedness identification strategies hold, it will be the case that
    \begin{align*}
        \E[Y_{t^*_1} \mid G=g] = \E[Y_{t^*_1} \mid G=g'] \ \text{for all $g,g' \in \mathcal{G}^*$}
    \end{align*}
    In other words, all close comparison groups should have the same average outcome in post-treatment periods.
\end{remark}

\begin{remark}[Single available comparison group]
    The discussion in this section has focused on the case where $|\bar{\mathcal{G}}|>1$, i.e., where there are multiple possible comparison groups, but the arguments can continue to go through when $|\bar{\mathcal{G}}|=1$; in this case, our approach amounts to simply checking whether or not the one available comparison group is a close comparison group or not.  The main value of having more comparison groups is that the researcher has more opportunities to find close comparison groups among them.
\end{remark}

\subsection{Alternative approaches to forming the set of close comparison groups}

Above, we formed the set of close comparison groups on the basis of having the same distribution of outcomes as the treated group in the period immediately preceding the treatment.  It is possible to form a set of close comparison groups using different criteria, and these different criteria lead to different robustness properties.  We consider several different alternative definitions of close comparison groups below.

\subsubsection{Single Pre-Treatment Period}

\noindent \textit{\bfseries One-period means:} $\mathcal{G}^*_{\text{mean}}(1) := \big\{ g \in \bar{\mathcal{G}} : \E[Y_{t=1} \mid G=1] = \E[Y_{t=1} \mid G=g]\big\}$

\bigskip

$\mathcal{G}^*_{\text{mean}}(1)$ collects all of the possible comparison groups that have the same mean pre-treatment outcome as the treated group.  This is a weaker requirement than having the same distribution of outcomes in the pre-treatment period.  Thus, $\mathcal{G}^* \subseteq \mathcal{G}^*_{\text{mean}}(1)$.

\subsubsection{Multiple Pre-Treatment Periods}

For some $S \in \{1,\ldots,t^*_1-1\}$, we consider the following alternative definitions of close comparison groups.

\bigskip

\noindent \textit{\bfseries Multiple-period means:} $\mathcal{G}^*_{\text{mean}}(S) := \Big\{ g \in \bar{\mathcal{G}} : \E[Y_{t=t^*_1-j} \mid G=1] = \E[Y_{t^*_1-j} \mid G=g] \text{ for } j=1,\ldots, S \Big\}$

\bigskip

\noindent \textit{\bfseries Multi-period distributions:} $\mathcal{G}^*_{\text{dist}}(S) := \Big\{ g \in \bar{\mathcal{G}} : \big(Y_{t^*_1-1}, \ldots, Y_{t^*_1-S} \mid G=1\big) \overset{d}{=} \big(Y_{t^*_1-1}, \ldots, Y_{t^*_1-S} \mid G=g\big) \Big\}$

\bigskip

$\mathcal{G}^*_{\text{mean}}(S)$ generalizes $\mathcal{G}^*_{\text{mean}}(1)$ so that a close comparison group must have the same mean of outcomes across $S$ pre-treatment periods, which is a more stringent criteria when $S>1$.  Thus, $\mathcal{G}^*_{\text{mean}}(S) \subseteq \mathcal{G}^*_{\text{mean}}(1)$.  The researcher can choose different values of $S$, up to $S=t^*_1-1$ (which would amount to using all available pre-treatment periods), and larger values of $S$ can only shrink the number of close comparison groups.

$\mathcal{G}^*_{\text{dist}}(S)$ is the most restrictive definition of a close comparison group that we consider in the paper.  It generalizes $\mathcal{G}^*$---instead of requiring that a group has the same distribution of outcomes in the period immediately preceding treatment, it requires that a group has the same joint distribution of pre-treatment outcomes across $S$ periods preceding treatment as the treated group has.  This leads to the smallest set of close comparison groups among those that we have considered.  Notice that $\mathcal{G}^*_{\text{dist}}(S) \subseteq \mathcal{G}^*$ and $\mathcal{G}^*_{\text{dist}}(S) \subseteq \mathcal{G}^*_{\text{mean}}(S)$.  In addition, similar to the discussion above, larger values of $S$ can only shrink the number of close comparison groups.

\subsubsection{Identification under alternative definitions of close comparison groups}

This section contains our second main result relating the five identification strategies discussed previously to the alternative definitions of a close comparison group discussed above.  Before stating this result, we define the following estimands
\begin{align*}
    \tau^g_{\text{mean}}(1) &:= \E[Y_{t^*_1}|G=1] - \E[Y_{t^*_1}|G=g] \qquad \text{ for $g \in \mathcal{G}^*_{\text{mean}}(1)$} \\
    \tau^g_{\text{mean}}(S) &:= \E[Y_{t^*_1}|G=1] - \E[Y_{t^*_1}|G=g] \qquad \text{ for $g \in \mathcal{G}^*_{\text{mean}}(S)$} \\
    \tau^g_{\text{dist}}(S) &:= \E[Y_{t^*_1}|G=1] - \E[Y_{t^*_1}|G=g] \qquad \text{ for $g \in \mathcal{G}^*_{\text{dist}}(S)$}
\end{align*}
and
\begin{align*}
    \tau_{\text{mean}}(1) & := \sum_{g \in \mathcal{G}^*_{\text{mean}}(1)} \tau^g_{\text{mean}}(1) \cdot w^g_{\text{mean}(1)}(g) \\
    \tau_{\text{mean}}(S) & := \sum_{g \in \mathcal{G}^*_{\text{mean}}(S)} \tau^g_{\text{mean}}(S) \cdot w^g_{\text{mean}(S)}(g) \\ %
    \tau_{\text{dist}}(S) & := \sum_{g \in \mathcal{G}^*_{\text{dist}}(S)} \tau^g_{\text{dist}}(S) \cdot w^g_{\text{dist}(S)}(g) %
\end{align*}

The next proposition collects a number of results on whether close comparison groups, defined in the ways discussed above, combined with particular identification strategies guarantee that the post-treatment comparison of means, $\E[Y_{t^*_1} \mid G=1] - \E[Y_{t^*_1} \mid G=g]$, where $g$ is a close comparison group, recovers $\ATT_{t^*_1}$. 
\begin{proposition}  \label{prop:alt-def-close-comparison-groups} For IFE, suppose that \Cref{ass:ife-rank} holds, and, for latent unconfoundedness, suppose that $S \geq L$.  Under \Cref{ass:sampling,ass:groups,ass:no-anticipation,ass:single-treated}, the following table indicates when a group $g$ satisfying a certain definition of being a close comparison group implies that $\E[Y_{t^*_1} \mid G=1] - \E[Y_{t^*_1} \mid G=g] = \ATT_{t^*_1}$ under different identification strategies.

    \bigskip

    \begin{tabular}{lccc}
    \toprule
     & \shortstack{$g \in \mathcal{G}^*_{\text{mean}}(1)$ \\[4pt] $\implies \tau^g_{\text{mean}}(1) = \ATT_{t^*_1}$} & \shortstack{$g \in \mathcal{G}^*_{\text{mean}}(S)$ \\[4pt] $\implies \tau^g_{\text{mean}}(S) = \ATT_{t^*_1}$} & \shortstack{$g \in \mathcal{G}^*_{\text{dist}}(S)$ \\[4pt] $\implies \tau^g_{\text{dist}}(S) = \ATT_{t^*_1}$} \\
    \midrule
    \textbf{DiD} & \greencheck & \greencheck & \greencheck \\
    \textbf{CiC} & \redx       & \redx       & \greencheck  \\
    \textbf{LOU} & \redx       & \redx       & \greencheck \\
    \textbf{IFE} & \redx       & \greencheck & \greencheck \\
    \textbf{Lat.~Unc.} & \redx & \redx       & \greencheck \\
    \bottomrule
    \end{tabular}
\end{proposition}

\bigskip

The proof of \Cref{prop:alt-def-close-comparison-groups} is provided in Appendix \ref{app:proofs}.  There are several main takeaways from this proposition.  First, among the identification strategies that we consider, a group $g$ having the same mean outcome in a single pre-treatment period only guarantees that $\tau_{\text{mean}}(1) = \ATT_{t^*_1}$ for difference-in-differences.  The second column indicates that having the same mean outcomes across multiple pre-treatment periods is additionally sufficient to guarantee that $\tau_{\text{mean}}(S) = \ATT_{t^*_1}$ for interactive fixed effects, as long as a rank condition related to the number of available pre-treatment time periods holds.  On the other hand, it is not sufficient for change-in-changes, lagged outcomes, or latent unconfoundedness.  The intuition for this result is that DiD and IFE models are linear in unobserved heterogeneity, so groups with the same pre-treatment means of outcomes also have the same means of unobserved heterogeneity, as long as there are enough pre-treatment periods.  In contrast, change-in-changes, lagged outcome unconfoundedness, and latent unconfoundedness are nonlinear in unobserved heterogeneity, so having the same pre-treatments means of outcomes is not sufficient to guarantee that the distribution of unobserved heterogeneity (or lagged outcomes) is the same between groups. Finally, the last column shows that having the same joint distribution of outcomes across multiple pre-treatment periods is sufficient to guarantee that $\tau_{\text{dist}}(S) = \ATT_{t^*_1}$ for all five identification strategies considered here, as long as there are enough pre-treatment periods relative to the dimension of unobserved heterogeneity.  Essentially, exhibiting the same joint distribution of outcomes over several pre-treatment periods guarantees that groups have the same distribution of unobserved heterogeneity, even if the model is nonlinear in unobserved heterogeneity.  And having the same distribution of unobserved heterogeneity implies that post-treatment comparisons of means recover the $\ATT$, even for nonlinear models that can vary over time in complicated ways.

\subsection{Other identification strategies:}

Our discussion above has focused on five important identification strategies for causal inference with panel data; however, it is important to note that our arguments apply to other identification strategies as well.  We briefly mention two additional identification strategies here and conjecture that our arguments could be extended to others as well.

\bigskip

\noindent \textit{\bfseries Unit-specific linear trends: }
\begin{align} \label{eqn:linear-trends}
    Y_{it}(0) = \theta_t + \eta_i + \lambda_i t + e_{it} \quad \text{ with $\E[e_t \mid G=g] = 0$ for all $g \in \mathcal{G}$.}
\end{align}

\bigskip

\noindent \textit{\bfseries Dynamic panel data model: }
\begin{align} \label{eqn:dynamic-panel-model}
    Y_{it}(0) = \theta_t + \eta_i + \rho Y_{it-1}(0) + e_{it} \quad \text{ with $\E[e_t \mid G=g] = 0$ for all $g \in \mathcal{G}$.}
\end{align}

\bigskip

Unit-specific linear trends models are popular in empirical work and have been considered in \citet{heckman-hotz-1989,wooldridge-2005}.  Dynamic panel data models for untreated potential outcomes are somewhat less common in empirical work focused on causal inference (see \citet{marx-tamer-tang-2025} for a recent paper that considers this type of setting).  In the next proposition, we relate these models to our approach of finding close comparison groups.

\begin{proposition} \label{prop:linear-trends-and-dynamic-panel} Suppose that $S \geq 2$.  Under \Cref{ass:sampling,ass:groups,ass:no-anticipation,ass:single-treated}, the following table indicates when a group $g$ satisfying a certain definition of being a close comparison group implies that $\E[Y_{t^*_1} \mid G=1] - \E[Y_{t^*_1} \mid G=g] = \ATT_{t^*_1}$ under different identification strategies.

\bigskip

{ \setlength{\tabcolsep}{3pt} \small
\noindent \begin{tabular}{lcccc}
    \toprule
     & \shortstack{$g \in \mathcal{G}^*_{\text{mean}}(1)$ \\[4pt] $\implies \tau^g_{\text{mean}}(1) = \ATT_{t^*_1}$} & \shortstack{$g \in \mathcal{G}^*_{\text{dist}}(1)$ \\[4pt] $\implies \tau^g_{\text{dist}}(1) = \ATT_{t^*_1}$} & \shortstack{$g \in \mathcal{G}^*_{\text{mean}}(S)$ \\[4pt] $\implies \tau^g_{\text{mean}}(S) = \ATT_{t^*_1}$} & \shortstack{$g \in \mathcal{G}^*_{\text{dist}}(S)$ \\[4pt] $\implies \tau^g_{\text{dist}}(S) = \ATT_{t^*_1}$} \\
    \midrule
    \textbf{Linear Trends} & \redx & \redx & \greencheck & \greencheck \\
    \textbf{Dynamic Panel} & \redx & \redx & \greencheck  & \greencheck \\
    \bottomrule
\end{tabular}
}
\end{proposition}

\bigskip

The proof is provided in Appendix \ref{app:proofs}.  \Cref{prop:linear-trends-and-dynamic-panel} shows that both linear trends and dynamic panel models for untreated potential outcomes behave very similarly to interactive fixed effects models with $R=1$ (i.e., with one interactive fixed effect).  Thus, defining the set of close comparison groups based on two periods of balance with respect to means is sufficient for $\tau^g_{\text{mean}}(2) = \ATT_{t^*_1}$.

\section{Robust Identification with Time Homogeneity} \label{sec:time-homogeneity}

\Cref{rem:debiasing-pre} emphasized that all of the two-period identification strategies that we have considered involve a cross-group comparison de-biased by differences between the groups in pre-treatment periods (where the particular form of de-biasing depended on the identification strategy).  Then, we tried to find groups where the de-biasing was unnecessary because the groups were already sufficiently similar.

A different, but related, viewpoint for all of these identification strategies is that they can be viewed as a de-biased before-after comparison where the de-biasing is to try to account for time effects on untreated potential outcomes.  For example, for difference-in-differences:
\begin{align*}
    \ATT_{t^*_1} = \underbrace{\E[Y_{t^*_1} - Y_{t^*_1-1} \mid G=1]}_{\text{before-after}} - \underbrace{\E[Y_{t^*_1} - Y_{t^*_1-1} \mid G=g]}_{\text{time de-biasing}}
\end{align*}
where, as above, the exact form of time de-biasing depends on the identification strategy.  In line with our approach above, this suggests that an alternative way to provide an identification-strategy-robust approach is to look for periods where outcomes do not change over time for the comparison groups, which is what we pursue here.

We define ``periods where outcomes do not change over time for the comparison groups'' in terms of there being no change in the mean outcome for each group across time periods (though the mean outcome need not be the same for each group).  Consider
\begin{align*}
    t_\text{mean}^{\text{TH}} &:=\max \big\{ s \in \{t^*_1, \ldots, \T\} : \E[Y_{t=s'}(0) \mid G=g] = \E[Y_{t=t^*_1-1}(0) \mid G=g] \ \text{for all } s'\leq s \text{ and } g \in \bar{\mathcal{G}}\big\}
\end{align*}
which is the maximum post-treatment period for which the mean outcome did not change relative to the period immediately preceding treatment for any comparison group.  We follow the convention of setting $t_{\text{mean}}^{\text{TH}}$ to 0 if the criteria is not met for any period.

Next, we define the following sets of time periods:
\begin{align*}
    \mathcal{T}_{\text{mean}} &= \big\{ t \in \{t^*_1, \ldots, \T\} : t \leq t_{\text{mean}}^{\text{TH}} \big\}
\end{align*}
which is the set of post-treatment time periods for which mean time homogeneity holds.  Note that this set can be empty if there are no time periods where there are no observed periods in which time homogeneity holds.

For this section, we consider the following estimand:
\begin{align*}
    \tau^{\text{TH}}_t := \E[Y_t - Y_{t^*_1-1} \mid G=1]
\end{align*}
for some period $t \in \mathcal{T}_{\text{mean}}$.  $\tau^{\text{TH}}_t$ is the average change in outcomes from period $Y_{t^*_1-1}$, the period right before group 1 becomes treated, to period $t$ for the treated group.

The following proposition contains our main result on the robustness of $\tau^{\text{TH}}_t$ under mean time homogeneity.

\begin{proposition}  \label{prop:time-homogeneity} For IFE, suppose that \Cref{ass:ife-comparison-group-rank} holds.  Under \Cref{ass:sampling,ass:groups,ass:no-anticipation,ass:single-treated}, and if $\mathcal{T}_{\text{mean}} \neq \emptyset$, then, for any $t \in \mathcal{T}_{\text{mean}}$, the following table indicates when $\tau^{\text{TH}}_t = \ATT_t$ under different identification strategies:

    \bigskip

    \begin{center}
    \begin{tabular}{lcc}
    \toprule
    & $\tau^{\text{TH}}_t = \ATT_{t}$  \\
    \midrule
    \textbf{DiD} & \greencheck \\
    \textbf{CiC} & \redx \\
    \textbf{LOU} & \redx \\
    \textbf{IFE} & \greencheck \\
    \textbf{Lat.~Unc.} & \redx   \\
    \bottomrule
    \end{tabular}
    \end{center}

\end{proposition}

The proof is provided in Appendix \ref{app:proofs}. The proposition shows that difference-in-differences and interactive fixed effects identification strategies both imply that $\tau^{\text{TH}}_t = \ATT_t$ for any $t$ in the set of time periods where mean time homogeneity holds.  That mean time homogeneity is sufficient for $\tau^{\text{TH}}_t = \ATT_t$ under an IFE identification strategy is quite interesting, as typically it would not be possible to identify $\ATT_t$ under IFE with only two time periods, where quasi-differencing strategies would typically require at least $(R+1)$ pre-treatment periods.  However, under an IFE model for untreated potential outcomes, if we have access to a sufficient number of ``distinct'' comparison groups, then mean time homogeneity effectively implies that the factors, $F_t$, are constant over the relevant time periods, which then implies that untreated potential outcomes for the treated group are time homogeneous as well.

In contrast, change-in-changes, lagged outcome unconfoundedness, and latent unconfoundedness do not guarantee that $\tau^{\text{TH}}_t = \ATT_t$ even when there are no time effects for comparison groups.  The intuition for this result is that the nonlinearity of each of these identification strategies leads to mean time homogeneity not being sufficient to guarantee that untreated potential outcomes for the treated group are time homogeneous as well. That being said, \Cref{prop:lou-mean-time-homogeneity} in Appendix \ref{app:proofs} counterbalances this result, at least to some extent, by showing that mean time homogeneity can be sufficient for $\tau^{\text{TH}}_t = \ATT_t$ under lagged outcome unconfoundedness under a ``flexible'' parametric model for $\E[Y_{t}(0) \mid Y_{t^*_1-1}(0), G=g]$, where the amount of flexibility that can be allowed for is connected to the number of available comparison groups.

\section{Multiply Robust Identification}

As discussed above, comparisons (i)between close comparison groups and (ii) based on time homogeneity can both be used to identify $\ATT_{t^*_1}$.  In both cases, these comparisons are robust to multiple identification strategies. In some applications, it may be the case that the researcher has access to both close comparison groups and a setting where time homogeneity holds.  Given that the robustness is non-nested for each approach, it suggests combining these two approaches to provide an even more robust estimator of $\ATT_{t^*_1}$.  To do so, we consider the following estimand:
\begin{align*}
    \tau_\text{mr} := \E[Y_{t^*_1} \mid G=1] - \sum_{g \in \mathcal{G}^*} w_g \E[Y_{t^*_1} \mid G=g] - \left(\E[Y_{t^*_1-1} \mid G=1] - \sum_{g \in \mathcal{G}^*} w_g \E[Y_{t^*_1-1} \mid G=g]\right)
\end{align*}
$\tau_{\text{mr}}$ is the difference between the post-treatment comparison of means between the treated group and close comparison groups adjusted by the analogous pre-treatment comparison of means.  The following discussion illustrates that $\tau_{\text{mr}}$ identifies $\ATT_{t^*_1}$ if either close comparison groups identification or time homogeneity identification holds.

\bigskip

\noindent \underline{\textit{Case 1:}} Suppose that close comparison groups identification works, i.e., $\E[Y_{t^*_1}(0) \mid G=1] = \sum_{g \in \mathcal{G}^*} w_g \E[Y_{t^*_1}(0) \mid G=g]$, but time homogeneity does not hold, i.e., $\E[Y_{t^*_1}(0) \mid G=1] \neq \E[Y_{t^*_1-1}(0) \mid G=1]$.  Then,%
\begin{align*}
    \tau_\text{mr} &= \Big(\E[Y_{t^*_1}(1) \mid G=1] - \sum_{g \in \mathcal{G}^*} w_g \E[Y_{t^*_1}(0) \mid G=g]\Big) + \Big(\E[Y_{t^*_1-1}(0) \mid G=1] - \sum_{g \in \mathcal{G}^*} w_g \E[Y_{t^*_1-1}(0) \mid G=g]\Big) \\
    &= \ATT_{t^*_1} + 0 = \ATT_{t^*_1}
\end{align*}
where the first equality holds by writing $\tau_{\text{mr}}$ in terms of potential outcomes, and the second equality holds because close comparison groups holds---the first term is equal to $\ATT_{t^*_1}$ because close comparison groups holds, and the second term is equal to zero because, under any of our definitions of close comparison groups, the treated group and all comparison groups have the same mean outcome in the pre-treatment period $t^*_1-1$.

\bigskip

\noindent \underline{\textit{Case 2:}}  Suppose that time homogeneity identification holds, i.e., $\E[Y_{t^*_1}(0) \mid G=1] = \E[Y_{t^*_1-1}(0) \mid G=1]$, but close comparison groups does not hold, i.e., $\E[Y_{t^*_1}(0) \mid G=1] \neq \sum_{g \in \mathcal{G}^*} w_g \E[Y_{t^*_1}(0) \mid G=g]$.  Then, %
\begin{align*}
    \tau_\text{mr} &= \Big(\E[Y_{t^*_1}(1) \mid G=1] - \sum_{g \in \mathcal{G}^*} w_g \E[Y_{t^*_1}(0) \mid G=g]\Big) + \Big(\E[Y_{t^*_1-1}(0) \mid G=1] - \sum_{g \in \mathcal{G}^*} w_g \E[Y_{t^*_1-1}(0) \mid G=g]\Big) \\
    &= \Big(\E[Y_{t^*_1}(1) \mid G=1] - \E[Y_{t^*_1-1}(0) \mid G=1]\Big) - \sum_{g \in \mathcal{G}^*} w_g \Big( \E[Y_{t^*_1}(0) \mid G=g] - \E[Y_{t^*_1-1}(0) \mid G=g]\Big) \\
    &= \ATT_{t^*_1} + 0 = \ATT_{t^*_1}
\end{align*}
where the first equality holds by writing $\tau_{\text{mr}}$ in terms of potential outcomes, the second equality holds by rearranging terms, and the last line holds because time homogeneity holds---the first term is equal to $\ATT_{t^*_1}$ because time homogeneity holds, and the second term is equal to zero because time homogeneity holds among all comparison groups.

\section{Estimation and Inference}

In this section, we propose an estimator for $\ATT_{t^*_1}$.  The key challenge for our estimation procedure is that the set of close comparison groups is not known and must be estimated.  Given an estimate of this, we effectively then plug it into the identification formula from \Cref{thm:1} to obtain an estimate of $\ATT_{t^*_1}$.  We also propose a method for conducting valid inference, taking into account the uncertainty in choosing the set of close comparison groups.

Slightly abusing notation, in this section, we use the notation $\mathcal{G}^*$ to indicate the set of close comparison groups.  This is the same notation we used above for choosing the set of close comparison groups based on having the same distribution of outcomes in the period immediately before the treatment; however, the arguments in this section apply to any of the notions of forming close comparison groups earlier in the paper.  We consider the following estimator for $\ATT_{t^*_1}$. %
\begin{align*}
    \widehat{\ATT}_{t^*_1}(h_n) = \bar{Y}_{1,t^*_1} - \sum_{g \in \mathcal{G}} \widehat{w}_g(h_n) \cdot \bar{Y}_{g,t^*_1}
\end{align*}
where
\begin{align*}
    \bar{Y}_{g,t^*_1} := \frac{1}{n_g} \sum_{i=1}^n \indicator{G_i=g} Y_{it^*_1} \quad \text{and} \quad \widehat{w}_g(h_n) := \frac{\K\big(\hat{\d}_g/h_n\big) \cdot \widehat{p}_g}{\displaystyle \sum_{g' \in \mathcal{G}} \K\big(\hat{\d}_{g'}/h_n\big) \cdot \widehat{p}_{g'}}
\end{align*}
and where $\K$ is a compactly supported kernel function, $h_n$ is a bandwidth, and $\hat{\d}_g$ is an estimate of the distance between group 1 and group $g$.  Together, having a compactly supported kernel and the bandwidth result in a form of model selection, where only certain groups receive any positive weight.
We use the notation $\widehat{\mathcal{G}}^*$ to denote the implied estimator of the set of close comparison groups.  An important special case is when $\K$ is a uniform kernel, i.e., $\K(u) = \indicator{|u| \leq 1}$, in which case the estimator reduces to
\begin{align*}
    \widehat{\ATT}_{t^*_1} = \frac{1}{n_1} \sum_{i=1}^n \indicator{G_i=1} Y_{it^*_1} - \frac{1}{n_{\widehat{\mathcal{G}}^*}} \sum_{i=1}^n \indicator{G_i \in \widehat{\mathcal{G}}^*} Y_{it^*_1}
\end{align*}
where %
$n_1$ is the number of units in the treated group, and $n_{\widehat{\mathcal{G}}^*}$ is the number of units in the estimated close comparison groups.  This estimator is simply the difference in average outcomes for the treated group relative to the average outcomes in the estimated close comparison groups in the post-treatment period.

The key estimation task is to measure how ``close'' each untreated group $g$ is to the treated group in pre-treatment periods. Next, we give two leading examples for distance metrics, though we note that our theory applies to any distance metric that satisfies fairly minimal properties discussed below, and the appropriate choice of distance metric mainly comes down to the definition of close comparison groups that the researcher chooses.  For example, if the researcher chooses close comparison groups based on having the same distribution of outcomes as the treated group in the period immediately preceding treatment, then a natural distance metric is
\begin{align*}
    \d^{\text{dist}}_g := \frac{ \int_{0}^{1} \big( \Q_{1}(u) - \Q_{g}(u) \big)^2 \, du }{ \sqrt{(\sigma_1^2 + \sigma_g^2)/2} }
\end{align*}
where $\Q_g(u) := \Q_{Y_{t^*_1-1}|G=g}(u)$ is the quantile function of the outcome distribution for group $g$ in the period immediately preceding treatment, and $\sigma_g^2 := \Var(Y_{t^*_1-1} \mid G=g)$.  $\d^{\text{dist}}_g$ is a standardized 2-Wasserstein distance between the treated and group $g$ pre-treatment outcome distributions.  The distance is zero if and only if the pre-treatment outcome distributions match, i.e., if group $g$ belongs to the set of close comparison groups.

Similarly, if close comparison groups are based on having the same mean outcomes as the treated group across several pre-treatment periods, then a natural distance metric is
\begin{align*}
    \quad \d^{\text{mean}}_g := \sqrt{ (\bm{\mu}_g - \bm{\mu}_1)' \Big((\bm{\Sigma}_1 + \bm{\Sigma}_g)/2\Big)^{-1} (\bm{\mu}_g - \bm{\mu}_1) }.
\end{align*}
where $\bm{\mu}_g := (\E[Y_{t^*_1-1} \mid G=g], \ldots, \E[Y_{t^*_1-S} \mid G=g])'$ is the vector of mean outcomes for group $g$ for the $S$ periods preceding treatment, and $\bm{\Sigma}_g$ is the covariance matrix of outcomes for group $g$ across the same $S$ periods.  $\d^{\text{mean}}_g$ is a multivariate standardized distance between the treated and group $g$ pre-treatment mean outcome vectors---it reduces to the standardized difference in \citet{imbens-rubin-2015} in the case where $S=1$.  The distance is zero if and only if the pre-treatment mean outcome vectors match, i.e., if group $g$ belongs to the set of close comparison groups.

Whatever the distance metric that is ultimately used, it needs to be estimated.  For the examples given above, natural estimators are
\begin{align*}
    \widehat{\d}^{\text{dist}}_g &:= \frac{ \displaystyle \frac{1}{J}\sum_{j=1}^J \big( \widehat{\Q}_{1}(u_j) - \widehat{\Q}_{g}(u_j) \big)^2 }{ \sqrt{(\widehat{\sigma}_1^2 + \widehat{\sigma}_g^2)/2} } \quad \text{and} \quad
    \widehat{\d}^{\text{mean}}_g := \sqrt{ (\widehat{\bm{\mu}}_g - \widehat{\bm{\mu}}_1)' \Big((\widehat{\bm{\Sigma}}_1 + \widehat{\bm{\Sigma}}_g)/2\Big)^{-1} (\widehat{\bm{\mu}}_g - \widehat{\bm{\mu}}_1) }
\end{align*}
where $\mathcal{U} = \{u_1, \ldots, u_J\}$ is an equally spaced grid on $(0,1)$, $\widehat{\Q}_g(u)$ is the empirical quantile function of the outcome distribution for group $g$ in the period immediately preceding treatment, $\widehat{\sigma}_g^2$ is the sample variance of outcomes for group $g$ in the same period, $\widehat{\bm{\mu}}_g$ is the vector of sample mean outcomes for group $g$ across the $S$ periods preceding treatment, and $\widehat{\bm{\Sigma}}_g$ is the sample covariance matrix of outcomes for group $g$ across the same $S$ periods.

\bigskip

We make the following additional assumptions for estimation and inference.

\begin{assumption}[Regularity Conditions] \label{ass:est-regularity} \
    \begin{itemize}
        \item [(i)] $\mathcal{G}^*$ is non-empty
        \item [(ii)] $p_g > 0$ for all $g \in \mathcal{G}$.
        \item [(iii)] $\E\big[ || \mathbf{Y} ||^2 \big] < \infty$ where $\mathbf{Y} = (Y_1, \ldots, Y_\T)'$.
    \end{itemize}
\end{assumption}

\begin{assumption}[Group Separation] \label{ass:est-group-separation} $\d_g = 0$ for all $g \in \mathcal{G}^*$.  Let $\d_{\text{min}} = \min_{g \notin \mathcal{G}^*} d_g$.  $\d_{\text{min}} > 0$.
\end{assumption}

\begin{assumption}[Bandwidth] \label{ass:est-bandwidth}
    The bandwidth $h_n$ satisfies $h_n \downarrow 0$ and $\sqrt{n_{\mathrm{min}}} \, h_n^2 \to \infty$, as $n_{\mathrm{min}} \rightarrow \infty$.
\end{assumption}

\begin{assumption}[Kernel] \label{ass:est-kernel}
    The kernel function $\K : \mathbb{R} \to \mathbb{R}$ satisfies:
    \begin{itemize}
        \item[(i)] $\K$ is symmetric: $\K(u) = \K(-u)$ for all $u \in \mathbb{R}$;
        \item[(ii)] $\K$ has compact support: $\K(u) = 0$ for $|u| > 1$;
        \item[(iii)] $\K(0) > 0$;
        \item[(iv)] $\K$ is twice continuously differentiable in a neighborhood of $0$;
        \item[(v)] $\K$ is bounded.
    \end{itemize}
\end{assumption}

\Cref{ass:est-regularity} provides several mild regularity conditions.  \Cref{ass:est-group-separation} says that, in population, the $\d_g$ should be equal to zero for close comparison groups and be non-zero for groups that are not close comparison groups.

\Cref{ass:est-bandwidth} provides the needed properties of the bandwidth.  It should converge to zero, but not too fast relative to the sample size.  The condition that $\sqrt{n_{\text{min}}} h_n^2 \rightarrow \infty$ can be weakened to $\sqrt{n_\text{min}} h_n \rightarrow \infty$ to establish consistency, but we use the stronger condition to establish asymptotic normality---it can be considered as a form of undersmoothing where we effectively impose a more stringent requirement on the allowable distance between group $g$ and the treated group in pre-treatment periods for $g$ to be considered a close comparison group.

\Cref{ass:est-kernel}, particularly the part about the kernel having compact support, rules out some kernels such as the Gaussian kernel, but the assumption allows for a uniform kernel as well as other kernels with compact support, such as an Epanechnikov kernel.  Moreover, using a compactly supported kernel provides a natural way to define $\widehat{\mathcal{G}}^* = \Big\{ g \in \mathcal{G} : \K\big(\widehat{\d}_g / h_n\big) > 0 \Big\}$, i.e., the estimated set of close comparison groups is given by all the groups having $\K\big(\widehat{\d}_g/h\big) > 0$.  We use \Cref{ass:est-kernel}(iv) to establish asymptotic normality of our estimator, but the differentiability requirement can be relaxed to continuity at $u=0$ for consistency.

\bigskip

\begin{assumption}[Asymptotic linearity of distance estimators]\label{ass:dist-if}
For each $g \in \mathcal{G}$, the distance estimator $\widehat \d_g$ admits an asymptotically linear representation
\begin{align*}
\sqrt{n_g}\,(\widehat{\d}_g - \d_g)
= \frac{1}{\sqrt{n_g}} \sum_{i=1}^n \indicator{G_i=g}\,\psi_{\d,g}(W_i) + o_p(1),
\end{align*}
where $W_i$ denotes the observed data for unit $i$, $\E[\psi_{\d,g}(W)\mid G=g]=0$, and $\Var\big(\psi_{\d,g}(W)\mid G=g\big)<\infty$.
\end{assumption}

\Cref{ass:dist-if} imposes a high-level condition on the distance metric and its estimator.  This assumption is likely to be weak in practice as all common distance metrics should satisfy this condition, though it may require additional regularity conditions.

\bigskip 

The next proposition provides our first main estimation result.

\begin{proposition} \label{prop:model-selection-consistency}
    Under \Cref{ass:sampling,ass:groups,ass:no-anticipation,ass:single-treated,ass:est-regularity,ass:est-group-separation,ass:est-bandwidth,ass:est-kernel,ass:dist-if}, the estimator $\widehat{\mathcal{G}}^*$ of $\mathcal{G}^*$ is consistent. That is, $\P( \widehat{\mathcal{G}}^* = \mathcal{G}^*) \to 1$ as $n_{\mathrm{min}} \to \infty$.
\end{proposition}
The proof is provided in \Cref{app:proofs}.  \Cref{prop:model-selection-consistency} says that our implied estimator of the close comparison groups, $\widehat{\mathcal{G}}^*$, is consistent for the true set of close comparison groups, $\mathcal{G}^*$, under the assumptions discussed above.

Next, we move towards establishing the asymptotic distribution of our estimator.  Towards this end, consider
\begin{align*}
    \oracle{\ATT}_{t^*_1} = \bar{Y}_{1,t^*_1} - \sum_{g \in \mathcal{G}^*} \frac{\widehat{p}_g}{ \sum_{g' \in \mathcal{G}^*}\widehat{p}_{g'} } \cdot \bar{Y}_{g,t^*_1}
\end{align*}
which is the oracle estimator of $\ATT_{t^*_1}$ that knows the true set of close comparison groups $\mathcal{G}^*$.

Finally, the following result provides the limiting distribution of our estimator.
\begin{theorem} \label{thm:asymptotic-normality} Under \Cref{ass:sampling,ass:groups,ass:no-anticipation,ass:single-treated,ass:est-regularity,ass:est-group-separation,ass:est-bandwidth,ass:est-kernel,ass:dist-if},
    \begin{align*}
        \sqrt{n}\left( \widehat{\ATT}_{t^*_1}(h_n) - \ATT_{t^*_1}\right) &= \frac{1}{\sqrt{n}} \sum_{i=1}^n \psi(W_i) + o_p(1) \\
    & \xrightarrow{d} \N(0,V)
    \end{align*}
    where
    \begin{align*}
    & \psi(W) = \frac{\indicator{G = 1}}{p_1} \big(Y_{t^*_1} - \E[Y_{t^*_1} \mid G=1] \big) - \frac{\indicator{G \in \mathcal{G}^*}}{p_{\mathcal{G}^*}} \big(Y_{t^*_1} - \E[Y_{t^*_1} \mid G \in \mathcal{G}^*] \big) \\[4pt]
    \text{and } & V = \E\big[ \psi(W)^2 \big] = \frac{\Var(Y_{t^*_1} \mid G=1)}{p_1} + \frac{\Var(Y_{t^*_1} \mid G \in \mathcal{G}^*)}{p_{\mathcal{G}^*}}
    \end{align*}
\end{theorem}
The proof is provided in \Cref{app:proofs}.  Our argument to show this result proceeds by showing that, under the conditions discussed above, our estimator is asymptotically equivalent to the oracle estimator above where $\mathcal{G}^*$ is known.  Then, our estimator amounts to the difference in average outcomes between units in the treated group and units in any of the treated group's close comparison groups.

\section{Extensions} \label{sec:more-extensions}

This section discusses several extensions to our main framework. We first show how to incorporate observed covariates. %
Second, we discuss how to handle settings with multiple treated groups and variation in treatment timing.

\subsection{Covariates} \label{subsec:covariates}

All of the arguments discussed above can be extended to hold conditional on covariates.  For example, for some $x \in \text{support}(X)$, we can define
\begin{align*}
    \mathcal{G}^*(x) := \big\{g \in \bar{\mathcal{G}} \bigm| (Y_{t^*_1-1} \mid X=x, G=1) \overset{d}{=} (Y_{t^*_1-1} \mid X=x, G=g) \big\}
\end{align*}
where $\mathcal{G}^*(x)$ is the conditional-on-covariates version of $\mathcal{G}^*$ from earlier in the paper, which is based on having the same distribution of outcomes in the period immediately preceding treatment; notice that we could define alternative versions of conditional-on-covariates close comparison groups based on means and/or multiple periods.  Then, defining
\begin{align*}
    \tau_g(x) &:= \E[Y_{t^*_1}|X=x,G=1] - \E[Y_{t^*_1} \mid X=x, G=g] \quad \text{for } g \in \mathcal{G}^*(x), \\[4pt]
    \text{and} \quad \tau(x) &:= \sum_{g \in \mathcal{G}^*(x)} \tau_g(x) w_g(x)
\end{align*}
for some weighting function $w_g(x)$---e.g., $w_g(x) = \P\big(G=g \mid X=x, G \in \mathcal{G}^*(x)\big)$.  Then, using the same arguments as discussed above, it is possible to show conditional-on-covariates versions of the results in \Cref{thm:1} and \Cref{prop:alt-def-close-comparison-groups,prop:linear-trends-and-dynamic-panel}, under which
\begin{align*}
    \tau(x) = \ATT(x)
\end{align*}
where $\ATT(x) := \E[Y_{t^*_1}(1) - Y_{t^*_1}(0) \mid X=x, G=1]$, is the conditional average treatment effect on the treated for $G=1$.  Moreover, it immediately follows that
\begin{align*}
    \E[\tau(X) \mid G=1] = \ATT
\end{align*}
There are two empirical issues worth briefly discussing.  First, the groups composing $\mathcal{G}^*(x)$ could change for different values of the covariates.  This does not create any issues except in the case where $\mathcal{G}^*(x)$ is empty for some values of the covariates, in which case $\tau(x)$ is not defined.  In this case, a natural approach is to redefine the target parameter to be
\begin{align*}
    \widetilde{\ATT} := \E[Y_{t^*_1}(1) - Y_{t^*_1}(0) \mid X \in \bar{\mathcal{X}}, G=1]
\end{align*}
where $\bar{\mathcal{X}} := \{ x \in \text{support}(X) : \mathcal{G}^*(x) \neq \emptyset \}$.  $\widetilde{\ATT}$ is the average treatment effect on treated across only values of the covariates for which we are able to find at least one close comparison group.  Re-defining the target parameter based on trimming out certain values of the covariates where there are no available comparison groups is conceptually similar to the idea of trimming when there are no available comparison units due to violations of common support in settings with a single comparison group (\citet{crump-hotz-imbens-mitnik-2009}).

Second, we leave dealing with high-dimensional covariates (or even continuous covariates) in estimation for future work; however, we note that our estimation results essentially apply immediately for low-dimensional, discrete covariates.

\subsection{Multiple Treated Groups and Variation in Treatment Timing} \label{subsec:multiple-treated}

Our approach extends straightforwardly to settings when multiple groups participate in the treatment and where treatment timing can vary across groups.  In these settings, it is useful to define the group-time average treatment effect, $ATT(g,t)$, the average treatment effect for group $g$ in time period $t$.  For each $(g,t)$ pair where $t \geq t^*_g$ (i.e., periods after group $g$ first becomes treated), comparison groups can be formed using the procedures outlined in the previous sections.  The key difference relative to the case emphasized above with a single treated group is that the set of available comparison groups may vary across different $(g,t)$ pairs.  For instance, in later time periods, the pool of not-yet-treated groups may shrink as more groups become treated.  Despite these time-varying comparison group structures, the core logic of our approach applies period-by-period and group-by-group: for each $ATT(g,t)$, one can construct a set of close comparison groups and compare average outcomes in time period $t$ for group $g$ relative to its close comparison groups.  Once the $ATT(g,t)$ parameters are estimated, they can be aggregated into more interpretable summary parameters.  Common aggregations include event study parameters, which trace out treatment effect dynamics relative to the time of treatment adoption, and an overall average treatment effect that averages across all groups and time periods.  The specific aggregation scheme depends on the ultimate target of the analysis; see \citet{callaway-santanna-2021} and \citet{callaway-2023} for detailed discussions of aggregation strategies and their interpretation.

\section{Conclusion}

This paper has considered an approach to panel data causal inference based on finding close comparison groups.  Rather than using panel data to adjust for differences in unobserved heterogeneity between treated and comparison groups, we use it to search for groups that are similar to the treated group in terms of their pre-treatment outcomes. We demonstrate that this approach is identification-strategy-robust, recovering the average treatment effect on the treated under many different non-nested panel data identification strategies.  Our approach also avoids the auxiliary assumptions that are required by the most common approaches to causal inference with panel data, such as difference-in-differences, change-in-changes, and lagged outcome unconfoundedness. The key insight is that comparison groups satisfying relatively simple criteria in pre-treatment periods provide valid comparisons under multiple identification strategies simultaneously, substantially increasing the credibility of causal inferences when such groups exist.

Our approach asks more of the research design by requiring the existence of close comparison groups, which is not necessary for most existing panel data approaches. This is the main downside of our approach.  However, when close comparison groups exist, researchers can avoid relying on functional form assumptions and identification strategies that may be difficult to justify.  While this restricts the set of applications where our method applies relative to approaches that impose stronger assumptions, it offers substantially more credible inferences in settings where multiple potential comparison groups are available and researchers have access to underlying microdata within each group.

{
\begingroup
\setstretch{.95}
\renewcommand*{\bibfont}{\small}
\setlength\bibitemsep{5pt}
\printbibliography
\endgroup
}

\appendix

\section{Proofs} \label{app:proofs}

\subsection{Close Comparison Group Proofs}

\begin{proof}[\textbf{Proof of \Cref{thm:1}}]
    We consider each identification strategy in turn.

    \bigskip

    \noindent \textit{ \underline{(i) Difference-in-Differences:}}
    For any group $g \in \mathcal{G}^*$
    \begin{align*}
        \tau_g &= \E[Y_{t^*_1} \mid G=1] - \E[Y_{t^*_1} \mid G=g] \\
        &= \E[Y_{t^*_1} \mid G=1] - \underbrace{\E[Y_{t^*_1-1} \mid G=1]} - \Big(\E[Y_{t^*_1}\mid G=g] - \underbrace{\E[Y_{t^*_1-1}\mid G=g]}\Big) \\
        &= \E[Y_{t^*_1} - Y_{t^*_1-1} \mid G=1] - \E[Y_{t^*_1}(0) - Y_{t^*_1-1}(0) \mid G=g] \\
        &= \E[Y_{t^*_1} - Y_{t^*_1-1} \mid G=1] - \E[Y_{t^*_1}(0) - Y_{t^*_1-1}(0) \mid G=1] \\
        &= \E[Y_{t^*_1}(1) - Y_{t^*_1}(0) \mid G = 1] \\
        &= \ATT_{t^*_1}
    \end{align*}
    where the first equality holds by the definition of $\tau_g$, the two underlined terms in the second line cancel because $g \in \mathcal{G}^*$, the third equality holds by rearranging terms,  the fourth equality holds by \Cref{ass:pt}, the fifth equality holds by rearraning, and the last equality holds by the definition of $\ATT_{t^*_1}$.

    This argument holds for all $g \in \mathcal{G}^*$, implying that $\tau_g = \tau = ATT_{t^*_1}$.

    \bigskip

    \noindent \textit{ \underline{(ii) Change-in-Changes:}} The proof uses the probability integral transform, i.e.:
    \begin{align} \label{eqn:pit}
        F_Y(Y) \sim \text{Uniform}[0,1]
    \end{align}
    and the inverse probability transform, i.e., for $U \sim \text{Uniform}[0,1]$:
    \begin{align} \label{eqn:inv-pit}
        Q_Y(U) \sim F_Y
    \end{align}

    For any group $g \in \mathcal{G}^*$,
    \begin{align*}
        \tau_g &= \E[Y_{t^*_1} \mid G=1] - \E[Y_{t^*_1}(0) \mid G=g] \\
        &= \E[Y_{t^*_1} \mid G=1] - \E[ Q_{Y_{t^*_1}(0) | G=g}(U) \mid G=g ] \\
        &= \E[Y_{t^*_1} \mid G=1] - \E\Big[ Q_{Y_{t^*_1}(0)|G=g} \circ F_{Y_{t^*_1-1}(0)|G=g}\big(Y_{t^*_1-1}(0)\big) \Bigm| G=g \Big] \\
        &= \E[Y_{t^*_1} \mid G=1] - \E\Big[ Q_{Y_{t^*_1}(0)|G=1} \circ F_{Y_{t^*_1-1}(0)|G=1}\big(Y_{t^*_1-1}(0)\big) \Bigm| G=g \Big] \\
        &= \E[Y_{t^*_1} \mid G=1] - \E\Big[ Q_{Y_{t^*_1}(0)|G=1} \circ F_{Y_{t^*_1-1}(0)|G=1}\big(Y_{t^*_1-1}(0)\big) \Bigm| G=1 \Big] \\
        &= \E[Y_{t^*_1} \mid G=1] - \E[ Q_{Y_{t^*_1}(0)|G=1}(U) \mid G=1] \\
        &= \E[Y_{t^*_1}(1) \mid G=1] - \E[Y_{t^*_1}(0) \mid G=1] \\
        &= \ATT_{t^*_1}
    \end{align*}
    where the first equality holds by the definition of $\tau_g$ and because group $g$ does not participate in the treatment, the second equality uses \Cref{eqn:inv-pit}, the third equality uses \Cref{eqn:pit}, the fourth equality holds by \Cref{ass:cic}, the fifth equality uses that $g \in \mathcal{G}^*$, the sixth equality uses \Cref{eqn:inv-pit}, the seventh equality rearranges terms, and the last equality holds by the definition of $\ATT_{t^*_1}$.

    \bigskip

    \noindent \textit{ \underline{(iii) Lagged Outcome Unconfoundedness:}}
    For any group $g \in \mathcal{G}^*$
    \begin{align*}
        \tau_g &= \E[Y_{t^*_1}\mid G=1] - \E[Y_{t^*_1}\mid G=g] \\
        &= \E[Y_{t^*_1}\mid G=1] - \E\Big[ \E\big[Y_{t^*_1}(0)\mid Y_{t^*_1-1}(0), G=g\big]\Big| G=g \Big] \\
        &= \E[Y_{t^*_1}\mid G=1] - \E\Big[ \E\big[Y_{t^*_1}(0)\mid Y_{t^*_1-1}(0), G=g\big]\Big| G=1 \Big] \\
        &= \E[Y_{t^*_1}\mid G=1] - \E\Big[ \E\big[Y_{t^*_1}(0)\mid Y_{t^*_1-1}(0), G=1\big]\Big| G=1 \Big] \\
        &= \E\big[Y_{t^*_1}(1) - Y_{t^*_1}(0) \mid G=1\big] \\
        &= \ATT_{t^*_1}
    \end{align*}
    where the first equality holds by the definition of $\tau_g$, the second equality holds because observed outcomes for group $g$ are untreated potential outcomes and by the law of iterated expectations, the third equality uses that $g \in \mathcal{G}^*$, the fourth equality holds by \Cref{eqn:lou}, the fifth equality rearranges terms, and the last equality holds by the definition of $\ATT_{t^*_1}$.  This argument holds for all $g \in \mathcal{G}^*$, implying that $\tau_g = \tau = ATT_{t^*_1}$.

    \bigskip

    \noindent \textit{\underline{(iv) Interactive Fixed Effects:}}
    We show that $g \in \mathcal{G}^*$ is not sufficient for $\tau_g = \ATT_{t^*_1}$ by counterexample.  Consider the case where $\T=2$ and $t^*_1=2$.  Suppose that $Y_{t=1} \mid G=1 \sim \N(1,1)$, and that $g \in \mathcal{G}^*$, hence, $Y_{t=1} \mid G=g \sim \N(1,1)$.  Furthermore, suppose that $R=1$, $F_1=1$, $F_2=2$, $\E[\eta \mid G=1]=1$, $\E[\lambda \mid G=1]=0$, $\E[\eta \mid G=g]=0$, and $\E[\lambda \mid G=g]=1$.  In this case, although group $g$ has the same distribution of untreated potential outcomes in the first period, $1=\E[Y_{t=2}(0) \mid G=1] \neq \E[Y_{t=2}(0) \mid G=g] = 2$, which implies that $\tau_g \neq \ATT$.

    \bigskip

    \noindent \textit{\underline{(v) Latent Unconfoundedness:}}
    The same counterexample for IFE works for latent unconfoundedness.
\end{proof}

\bigskip

\begin{proof}[\textbf{Proof of \Cref{prop:alt-def-close-comparison-groups}}]

    We organize the proof by each row in the table in \Cref{prop:alt-def-close-comparison-groups}.

    \bigskip

    \noindent \underline{\textit{(i) Difference-in-differences:}}
    For $g \in \mathcal{G}^*_{\text{mean}}(1)$, notice that the proof of part (i) of \Cref{thm:1} only relied on the means of pre-treatment outcomes being the same for the treated group as for group $g$; therefore, the same proof applies in this case.

    Next, if $g \in \mathcal{G}^*_{\text{mean}}(S)$, it also follows that $g \in \mathcal{G}^*_{\text{mean}}(1)$, which implies the result.

    Likewise, if $g \in \mathcal{G}^*_{\text{dist}}(S)$, it also follows that $g \in \mathcal{G}^*_{\text{mean}}(1)$, which implies the result.

    \bigskip

    \noindent \underline{\textit{(ii) Change-in-Changes:}}
    For $g \in \mathcal{G}^*_{\text{mean}}(1)$, consider the following counterexample.  Suppose that $Y_{t^*_1-1} \mid G=g \sim \text{Uniform}[0,1]$, $Y_{t^*_1-1} \mid G=1 \sim \text{Triangular}(0,1,0.5)$, and that $Y_{t^*_1}(0) \mid G=g \sim \chi^2_1$.  Then,
    \begin{align*}
        \tau^g_{\text{mean}}(1) &= \E[Y_{t^*_1} \mid G=1] - \E[Y_{t^*_1} \mid G=g] \\
        &= \ATT_{t^*_1} + \E[Y_{t^*_1}(0) \mid G=1] - \E[Y_{t^*_1}(0) \mid G=g] \\
        &= \ATT_{t^*_1} + \E\Big[ Q_{Y_{t^*_1}(0)|G=g}\big(\F_{Y_{t^*_1-1}(0)|G=g}(Y_{t^*_1-1})\big) \Bigm| G=1 \Big] - \E\Big[ Q_{Y_{t^*_1}(0)|G=g}\big(\F_{Y_{t^*_1-1}(0)|G=g}(Y_{t^*_1-1})\big) \Bigm| G=g \Big] \\
        &= \ATT_{t^*_1} + \E\Big[ Q_{Y_{t^*_1}(0)|G=g}(Y_{t^*_1-1}) \Bigm| G=1 \Big] - \E\Big[ Q_{Y_{t^*_1}(0)|G=g}(Y_{t^*_1-1}) \Bigm| G=g \Big] \\
        &\approx \ATT_{t^*_1} + 0.91 - 1 \\
        &\neq \ATT_{t^*_1}
    \end{align*}
    where the first equality holds by the definition of $\tau^g_{\text{mean}}(1)$, the second equality by adding and subtracting $\E[Y_{t^*_1}(0) \mid G=1]$, the third equality holds by the change-in-changes identification strategy, the simplification in the fourth equality holds by the setup of the counterexample (noting that it is easy to come up with many counterexamples given the previous expression and we choose a particular simple one), the fifth equality holds by approximating the quantiles of a $\chi^2_1$ random variable evaluated at draws from a triangular distribution as in the counterexample and because the mean of a $\chi^2_1$ random variable is 1.

    For $g \in \mathcal{G}^*(S)$, essentially the same as the immediately preceding counterexample can be used as it is invariant to the means of the outcomes being equal across additional pre-treatment periods.

    Finally, if $g \in \mathcal{G}^*_{\text{dist}}(S)$, it implies that $g \in \mathcal{G}^*$ from \Cref{thm:1}, and, therefore, the result follows immediately from the proof of part (ii) of \Cref{thm:1}.

    \bigskip

    \noindent \underline{\textit{(iii) Lagged Outcome Unconfoundedness:}}
    For $g \in \mathcal{G}^*_{\text{mean}}(1)$, consider the following counterexample.  Suppose that $Y_{t^*_1-1} \mid G=1 \sim \N(0,\sigma_1^2)$, $Y_{t^*_1-1} \mid G=g \sim \N(0,\sigma_g^2)$, and that $\E[Y_{t^*_1}(0) \mid Y_{t^*_1-1}(0), G] = Y_{t^*_1-1}^2$.  Then,
    \begin{align*}
        \tau^g_{\text{mean}}(1) &= \E[Y_{t^*_1} \mid G=1] - \E[Y_{t^*_1} \mid G=g] \\
        &= \ATT_{t^*_1} + \E[Y_{t^*_1}(0) \mid G=1] - \E[Y_{t^*_1}(0) \mid G=g] \\
        &= \ATT_{t^*_1} + \E\big[ \E[Y_{t^*_1}(0) \mid Y_{t^*_1-1}(0), G=1] \bigm| G=1 \big] - \E\big[ \E[Y_{t^*_1}(0) \mid Y_{t^*_1-1}(0), G=g] \bigm| G=g \big] \\
        &= \ATT_{t^*_1} + \E\big[ Y_{t^*_1-1}(0)^2 \bigm| G=1 \big] - \E\big[ Y_{t^*_1-1}(0)^2 \bigm| G=g \big] \\
        &= \ATT_{t^*_1} - \sigma_1^2 + \sigma_g^2 \\
        & \neq \ATT_{t^*-1}
    \end{align*}
    where the first equality holds by the definition of $\tau^g_{\text{mean}}(1)$, the second equality by adding and subtracting $\E[Y_{t^*_1}(0) \mid G=1]$, the third equality holds by the law of iterated expectations, the fourth equality holds by the setup of the counterexample, and the fifth and sixth equalities hold immediately.

    For $g \in \mathcal{G}^*(S)$, essentially the same as the immediately preceding counterexample can be used as it is invariant to the means of the outcomes being equal across additional pre-treatment periods.

    Finally, if $g \in \mathcal{G}^*_{\text{dist}}(S)$, it implies that $g \in \mathcal{G}^*$ from \Cref{thm:1}, and, therefore, the result follows immediately from the proof of part (iii) of \Cref{thm:1}.

    \bigskip

    \noindent \underline{\textit{(iv) Interactive Fixed Effects:}}
    For $g \in \mathcal{G}^*_{\text{mean}}(1)$, the same counterexample as in part (iv) of \Cref{thm:1} implies that this is not a sufficient condition to guarantee that $\tau^g_{\text{mean}}(1) = \ATT_{t^*_1}$.

    For $g \in \mathcal{G}^*_{\text{mean}}(S)$, notice that
    \begin{align*}
        \begin{pmatrix} \E[Y_{t^*_1-1} \mid G=1] \\ \vdots \\ \E[Y_{t^*_1-S} \mid G=1] \end{pmatrix} = \bm{\theta}^{\text{pre}(t^*_1)} + {\mathbf{F}^{\text{pre}(t^*_1)}}' \mathbf{\Lambda}_1 \quad \text{and} \quad \begin{pmatrix} \E[Y_{t^*_1-1} \mid G=g] \\ \vdots \\ \E[Y_{t^*_1-S} \mid G=g] \end{pmatrix} &= \bm{\theta}^{\text{pre}(t^*_1)} + {\mathbf{F}^{\text{pre}(t^*_1)}}' \mathbf{\Lambda}_g
    \end{align*}
    
    Then, because $S \geq (R+1)$ and by \Cref{ass:ife-rank}, it follows that $\mathbf{\Lambda}_1 = \mathbf{\Lambda}_g$.  Next, notice that
    \begin{align*}
        \tau^g_{\text{mean}}(S) &= \E[Y_{t^*_1} \mid G=1] - \E[Y_{t^*_1} \mid G=g] \\
        &= \ATT_{t^*_1} + \E[Y_{t^*_1}(0) \mid G=1] - \E[Y_{t^*_1}(0) \mid G=g] \\
        &= \ATT_{t^*_1} + \mathbf{\Lambda}_1'\mathbf{F}_{t^*_1} - \mathbf{\Lambda}_g'\mathbf{F}_{t^*_1} \\
        &= \ATT_{t^*_1}
    \end{align*}
    where the first equality holds by the definition of $\tau^g_{\text{mean}}(S)$, the second equality by adding and subtracting $\E[Y_{t^*_1}(0) \mid G=1]$, the third equality from \Cref{ass:ife-model}, and the last equality from the discussion above that showed that $g \in \mathcal{G}^*_{\text{mean}}(S)$ implied that $\mathbf{\Lambda}_1=\mathbf{\Lambda}_g$.

    Finally, if $g \in \mathcal{G}^*_{\text{dist}}(S)$, it also follows that $g \in \mathcal{G}^*_{\text{mean}}(S)$, which implies the result using the same argument as above.

    \bigskip

    \noindent \underline{\textit{(v) Latent Unconfoundedness:}}
    For $g \in \mathcal{G}^*_{\text{dist}}(S)$, notice that,
    \begin{align*}
        \tau^g_{\text{dist}}(S) &= \E[Y_{t^*_1} \mid G=1] - \E[Y_{t^*_1} \mid G=g] \\
        &= \ATT_{t^*_1} + \E[Y_{t^*_1}(0) \mid G=1] - \E[Y_{t^*_1}(0) \mid G=g] \\
        &= \ATT_{t^*_1} + \E[h_{t^*_1}(\xi) \mid G=1] - \E[h_{t^*_1}(\xi) \mid G=g] \\
        &= \ATT_{t^*_1}
    \end{align*}
    where the first equality comes from the definition of $\tau^g_{\text{dist}}(S)$, the second equality holds by adding and subtracting $\E[Y_{t^*_1}(0) \mid G=1]$, the third equality holds by \Cref{ass:latent-unconfoundedness}, %
    and the last equality also holds by \Cref{ass:latent-unconfoundedness-injectivity}---the last two terms are equal to each other because injectivity implies that groups having the same distribution of observed outcomes $(Y_{t^*_1-1}, \ldots,Y_{t^*_1-S})$ also have the same distribution of unobserved heterogeneity $\xi$.  
\end{proof}

\bigskip

\begin{proof}[\textbf{Proof of \Cref{prop:linear-trends-and-dynamic-panel}}]
    We organize the proof by each row in the table in \Cref{prop:linear-trends-and-dynamic-panel}.

    \bigskip

    \noindent \underline{\textit{Linear Trends:}}
    To start with, notice that the linear trends model for untreated potential outcomes in \Cref{eqn:linear-trends} is a special case of the interactive fixed effects model in \Cref{ass:ife-model} with $F_t = t$ for all $t$.  Moreover, for $g \in \mathcal{G}^*_{\text{mean}}(1)$ and $g \in \mathcal{G}^*_{\text{dist}}(1)$, notice that the same counterexample and proof as in \Cref{thm:1}(iv) applies here as it used involved $F_t=t$.

    Next, for $g \in \mathcal{G}^*_{\text{mean}}(S)$, let $\mathbf{Y}(0)^{\text{pre}(t^*_1,S)} = \big(Y_{t^*_1-S}(0), \dots, Y_{t^*_1-1}(0)\big)'$, $\bm{\theta}^{\text{pre}(t^*_1,S)} = \big(\theta_{t^*_1-S}, \dots, \theta_{t^*_1-1}\big)'$ and notice that
    \begin{align*}
        \E\big[\mathbf{Y}(0)^{\text{pre}(t^*_1,S)} \mid G=1\big] &=
        \bm{\theta}^{\text{pre}(t^*_1,S)} +
        \begin{pmatrix}
            1 & t^*_1-1 \\
            \vdots & \vdots \\
            1 & t^*_1 - S
        \end{pmatrix}
        \begin{pmatrix}
            \E[\eta \mid G=1] \\
            \E[\lambda \mid G=1]
        \end{pmatrix}, \text{ and } \\[4pt]
        \E\big[ \mathbf{Y}(0)^{\text{pre}(t^*_1,S)} \mid G=g \big] &= \bm{\theta}^{\text{pre}(t^*_1,S)} +
        \begin{pmatrix}
            1 & t^*_1-1 \\
            \vdots & \vdots \\
            1 & t^*_1 - S
        \end{pmatrix}
        \begin{pmatrix}
            \E[\eta \mid G=g] \\
            \E[\lambda \mid G=g]
        \end{pmatrix}
    \end{align*}
    The matrix $\begin{pmatrix}
            1 & t^*_1-1 \\
            \vdots & \vdots \\
            1 & t^*_1 - S
        \end{pmatrix}$ has full rank by construction, which implies that, if $g \in \mathcal{G}^*_{\text{mean}}(S)$, then
    \begin{align} \label{eqn:linear-trends-intermediate-step}
        \begin{pmatrix}
            \E[\eta \mid G=1] \\
            \E[\lambda \mid G=1]
        \end{pmatrix} = \begin{pmatrix}
            \E[\eta \mid G=g] \\
            \E[\lambda \mid G=g]
        \end{pmatrix}.
    \end{align}
    Next, notice that
    \begin{align*}
        \tau^g_{\text{mean}}(S) &= \E[Y_{t^*_1} \mid G=1] - \E[Y_{t^*_1} \mid G=g] \\
        &= \ATT_{t^*_1} + \E[Y_{t^*_1}(0) \mid G=1] - \E[Y_{t^*_1}(0) \mid G=g] \\
        &= \ATT_{t^*_1} + \Big( \theta_t + \E[\eta \mid G=1] + \E[\lambda \mid G=1] t^*_1 \Big) - \Big( \theta_t + \E[\eta \mid G=g] + \E[\lambda \mid G=g] t^*_1 \Big) \\
        &= \ATT_{t^*_1}
    \end{align*}
    where the first equality holds by the definition of $\tau^g_{\text{mean}}(S)$, the second equality holds by adding and subtracting $\E[Y_{t^*_1}(0) \mid G=1]$, the third equality holds from \Cref{eqn:linear-trends}, and the last equality holds by \Cref{eqn:linear-trends-intermediate-step}.

    Finally, if $g \in \mathcal{G}^*_{\text{dist}}(S)$, it implies that $g \in \mathcal{G}^*_{\text{mean}}(S)$, and, therefore, the result follows from the discussion above.

    \bigskip

    \noindent \underline{\textit{Dynamic Panel:}}
    For $g \in \mathcal{G}^*_{\text{mean}}(1)$, consider the following counterexample.  Suppose that $\theta_{t^*_1}=\theta_{t^*_1-1}=\theta_{t^*_1-2}=0$, $Y_{t^*_1-2}(0) \mid \eta, G=1 \sim \N(0,1)$, $Y_{t^*_1-2}(0) \mid \eta, G=g \sim \N(2,1)$, $\eta \mid G=1 \sim \N(1,1)$, $\eta \mid G=g \sim \N(0,1)$, $e_{t^*_1-1} \mid \eta, Y_{t^*_1-2}(0), G \sim \N(0,1)$, $e_{t^*_1} \mid \eta, Y_{t^*_1-1}(0), G \sim \N(0,1)$, and $\rho=0.5$.  Notice that $\E[\eta \mid G=1] = 1$, $\E[\eta \mid G=g] = 0$, $\E[Y_{t^*_1-2}(0) \mid G=1] = 0$, $\E[Y_{t^*_1-2}(0) \mid G=g] = 2$, and $\rho = 0.5$, which suffices for $g$ to be in $\mathcal{G}^*_{\text{mean}}(1)$.  Then, notice that
    \begin{align*}
        \tau^g_{\text{mean}}(1) &= \E[Y_{t^*_1} \mid G=1] - \E[Y_{t^*_1} \mid G=g] \\
        &= \ATT_{t^*_1} + \E[Y_{t^*_1}(0) \mid G=1] - \E[Y_{t^*_1}(0) \mid G=g] \\
        &= \ATT_{t^*_1} + \Big( \theta_t + \E[\eta \mid G=1] + \rho \E[Y_{t^*_1-1}(0) \mid G=1] \Big) - \Big( \theta_t + \E[\eta \mid G=g] + \rho \E[Y_{t^*_1-1}(0) \mid G=1] \Big) \\
        &= \ATT_{t^*_1} + \Big( 1 + 0.5\times 1\Big) - \Big(0 + 0.5 \times 1\Big) \\
        &\neq \ATT_{t^*_1}
    \end{align*}

    For $g \in \mathcal{G}^*_{\text{dist}}(1)$, we can use the same counterexample.  Notice that $Y_{t^*_1-1}(0) \mid G=1 \sim \N(1, 2.25) \overset{d}{=} Y_{t^*_1-1}(0) \mid G=g$, which implies that $g \in \mathcal{G}^*_{\text{dist}}(1)$ in the counterexample, and the same argument as above implies that $\tau^g_{\text{dist}}(1) \neq \ATT_{t^*_1}$.

    For $g \in \mathcal{G}^*_{\text{mean}}(S)$, we have that
    \begin{align} \label{eqn:dynamic-panel-intermediate-equation}
        \begin{pmatrix}
            \E[Y_{t^*_1-1}(0) \mid G=1] \\
            \E[Y_{t^*_1-2}(0) \mid G=1]
        \end{pmatrix}
        =
        \begin{pmatrix}
            \E[Y_{t^*_1-1}(0) \mid G=g] \\
            \E[Y_{t^*_1-2}(0) \mid G=g]
        \end{pmatrix}
    \end{align}
    Moreover, notice that
    \begin{align*}
        \E[Y_{t^*_1-1}(0) \mid G=1] &= \theta_{t^*_1-1} + \E[\eta \mid G=1] + \rho \E[Y_{t^*_1-2}(0) \mid G=1]
    \end{align*}
    which implies that
    \begin{align} \label{eqn:dynamic-panel-same-mean-eta}
        \E[\eta \mid G=1] &= \E[Y_{t^*_1-1}(0) \mid G=1] - \theta_{t^*_1-1} - \rho \E[Y_{t^*_1-2}(0) \mid G=1] \nonumber \\
        &= \E[Y_{t^*_1-1}(0) \mid G=g] - \theta_{t^*_1-1} - \rho \E[Y_{t^*_1-2}(0) \mid G=g] \nonumber \\
        &= \E[\eta \mid G=g]
    \end{align}
    where the first equality holds by \Cref{eqn:dynamic-panel-model}, the second equality holds by \Cref{eqn:dynamic-panel-intermediate-equation}, and the last equality holds by using \Cref{eqn:dynamic-panel-model} again.  Then,
    \begin{align*}
        \tau^g_{\text{mean}}(S) &= \E[Y_{t^*_1} \mid G=1] - \E[Y_{t^*_1} \mid G=g] \\
        &= \ATT_{t^*_1} + \E[Y_{t^*_1}(0) \mid G=1] - \E[Y_{t^*_1}(0) \mid G=g] \\
        &= \ATT_{t^*_1} + \Big( \theta_t + \E[\eta \mid G=1] + \rho \E[Y_{t^*_1-1}(0) \mid G=1] \Big) - \Big( \theta_t + \E[\eta \mid G=g] + \rho \E[Y_{t^*_1-1}(0) \mid G=1] \Big) \\
        &= \ATT_{t^*_1}
    \end{align*}
    where the first equality holds by the definition of $\tau^g_{\text{mean}}(S)$, the second equality holds by adding and subtracting $\E[Y_{t^*_1}(0) \mid G=1]$, the third equality holds by \Cref{eqn:dynamic-panel-model}, and the last equality holds by \Cref{eqn:dynamic-panel-intermediate-equation,eqn:dynamic-panel-same-mean-eta}.

    Finally, if $g \in \mathcal{G}^*(S)$, it implies that $g \in \mathcal{G}^*_{\text{mean}}(S)$, and, therefore, the result follows from the discussion above.
\end{proof}

\subsection{Time Homogeneity Proofs}

\begin{proof}[\textbf{Proof of \Cref{prop:time-homogeneity}}]
    We organize the proof by each row in the table in \Cref{prop:time-homogeneity}.

    \bigskip

    \noindent \underline{\textit{(i) Difference-in-Differences:}}
    Notice that, for any group $g \in \bar{\mathcal{G}}$,
    \begin{align*}
        \tau^{\text{TH}}_t &= \tau^{\text{TH}}_t - \E[Y_t(0) - Y_{t^*_1-1}(0) \mid G=g] \\
        &= \tau^{\text{TH}}_t - \E[Y_t(0) - Y_{t^*_1-1}(0) \mid G=1] \\
        &= \E[Y_t(1) - Y_t(0) \mid G=1] \\
        &= \ATT_t
    \end{align*}
    where the first equality holds because $t \in \mathcal{T}_{\text{mean}}$, the second equality holds by \Cref{ass:pt}, the third equality holds by adding and subtracting $\E[Y_{t^*_1-1}(0) \mid G=1]$ and by the definition of $\tau^{\text{TH}}_t$, and the last equality holds by the definition of $\ATT_t$.

    \bigskip

    \noindent \underline{\textit{(ii) Change-in-Changes:}}
    
    Consider the following counterexample.  For $g = 1, \ldots, |\mathcal{G}|$, suppose that $Y_{t^*_1-1}(0)\mid G=g \sim \N(g, 1)$.  We suppose that, for all groups, $Q_{Y_t(0)|G=g} \circ F_{Y_{t^*_1-1}(0)|G=g}(y) = y + \varepsilon\, q(y)$, where $\varepsilon$ is small and $q$ is a weighted sum of normal density functions placed at the group means:
    \begin{align*}
        q(y) = c_1 \phi(y - 1) + \cdots c_{|\mathcal{G}|} \phi(y - |\mathcal{G}|).
    \end{align*}
    where we set $c_1=1$, but allow $c_g, g=2,\ldots,|\mathcal{G}|$ to be constructed in order to allow for mean time homogeneity.  In particular, in order for mean time homogeneity to hold for all comparison groups, it must be the case that $\int q(y)\, f_g(y)\, dy = 0$ for all $g \in \bar{\mathcal{G}}$, where $f_g(y) = \phi(y-g)$ denotes the density of $Y_{t^*_1-1}(0) \mid G=g$.  Specifically, let
    \begin{align*}
        \mathbf{A} = \begin{bmatrix}
            \int \phi(y-g)\phi(y-1)\, dy & \int \phi(y-g) \phi(y-2)\, dy & \cdots & \int \phi(y-g)\phi(y-|\mathcal{G}|)\, dy
        \end{bmatrix}_{g \in \bar{\mathcal{G}}}
    \end{align*}
    which is a $|\mathcal{\bar{G}}| \times |\mathcal{G}|$ matrix, so that mean time homogeneity amounts to $\mathbf{A c} = 0$ where $\mathbf{c} = (c_1, \ldots, c_{|\mathcal{G}|})'$.  Since $|\mathcal{\bar{G}}| < |\mathcal{G}|$.  This is a system of $|\bar{\mathcal{G}}|$ linear equations with $|\bar{\mathcal{G}}|$ unknowns, which we can re-write as $\mathbf{A}_{-1} \mathbf{c}_{-1} = -\mathbf{A}_1$ where $\mathbf{A}_{-1}$ is the matrix $\mathbf{A}$ with the first column removed, $\mathbf{c}_{-1} = (c_2, \ldots, c_{|\mathcal{G}|})'$, and $\mathbf{A}_1$ is the first column of $\mathbf{A}$.  Then, if we set $\mathbf{c}_{-1} = - (\mathbf{A}_{-1})^{-1} \mathbf{A}_1$, then we have found a $q$ that satisfies mean time homogeneity for all comparison groups.

    Next, suppose that $f_1$ is not a linear combination of $f_g$'s, $g=2,\ldots,|\mathcal{G}|$.  Notice that $\E[Y_{t^*_1}(0) - Y_{t^*_1-1}(0) \mid G=1] = \varepsilon \int q(y)\, f_1(y)\, dy = \varepsilon (1 + c_2 \int \phi(y-1)\phi(y-2)\, dy + \cdots + c_{|\mathcal{G}|} \int \phi(y-1)\phi(y-|\mathcal{G}|)\, dy)$, which is generically nonzero since $c_1=1$ and the other terms are small (due to the small overlap between normal densities with different means).  Thus, $\tau^{\text{TH}}_t \neq \ATT_t$ even though mean time homogeneity holds for all comparison groups.

    To give a specific counterexample, consider $|\mathcal{G}|=3$ with two comparison groups $g \in \{2,3\}$). Then $f_1(y) = \phi(y-1)$, $f_2(y) = \phi(y-2)$, and $f_3(y) = \phi(y-3)$. The matrix $\mathbf{A}$ is $2 \times 3$:
    \begin{align*}
        \mathbf{A} = \begin{bmatrix}
            \int \phi(y-2)\phi(y-1)\, dy & \int \phi(y-2)\phi(y-2)\, dy & \int \phi(y-2)\phi(y-3)\, dy \\
            \int \phi(y-3)\phi(y-1)\, dy & \int \phi(y-3)\phi(y-2)\, dy & \int \phi(y-3)\phi(y-3)\, dy
        \end{bmatrix}.
    \end{align*}
    Using the fact that $\int \phi(y-j)\phi(y-k)\, dy = \frac{1}{\sqrt{4\pi}} \exp\left(-\frac{(j-k)^2}{4}\right)$, we have:
    \begin{align*}
        \mathbf{A} = \frac{1}{\sqrt{4\pi}} \begin{bmatrix}
            e^{-1/4} & 1 & e^{-1/4} \\
            e^{-1} & e^{-1/4} & 1
        \end{bmatrix}.
    \end{align*}
    With $c_1=1$, the system $\mathbf{A}_{-1} \mathbf{c}_{-1} = -\mathbf{A}_1$ becomes:
    \begin{align*}
        \frac{1}{\sqrt{4\pi}} \begin{bmatrix}
            1 & e^{-1/4} \\
            e^{-1/4} & 1
        \end{bmatrix} \begin{bmatrix} c_2 \\ c_3 \end{bmatrix} = -\frac{1}{\sqrt{4\pi}} \begin{bmatrix} e^{-1/4} \\ e^{-1} \end{bmatrix}.
    \end{align*}
    Solving yields $c_2 \approx -0.93$ and $c_3 \approx -0.15$. For the treated group:
    \begin{align*}
        \E[Y_{t^*_1}(0) - Y_{t^*_1-1}(0) \mid G=1] &= \varepsilon \int q(y)\, \phi(y-1)\, dy \\
        &= \varepsilon \left( \int \phi(y-1)^2\, dy + c_2 \int \phi(y-1)\phi(y-2)\, dy + c_3 \int \phi(y-1)\phi(y-3)\, dy \right) \\
        &= \frac{\varepsilon}{\sqrt{4\pi}} \left( 1 - 0.93 \cdot e^{-1/4} - 0.15 \cdot e^{-1} \right) \\
        &\approx \varepsilon \times 0.062 \neq 0.
    \end{align*}
    Thus, $\tau^{\text{TH}}_t \neq \ATT_t$ even though mean time homogeneity holds for both comparison groups.

    \bigskip

    \noindent \underline{\textit{(iii) Lagged Outcome Unconfoundedness:}}
    Suppose that
    \begin{align*}
        \E[Y_{t^*_1}(0) \mid Y_{t^*_1-1}(0), G=g] = \psi\big(Y_{t^*_1-1}(0)\big)'\beta
    \end{align*}
    for $\psi(y)$ a $J$-dimensional vector of basis functions. Further, let
    \begin{align*}
        \mathbf{\Psi} =  \begin{bmatrix} \E\big[\psi(Y_{t^*_1-1}(0)) \mid G=g\big]' \end{bmatrix}_{g \in \bar{\mathcal{G}}}
    \end{align*}
    which is a $|\bar{\mathcal{G}}| \times J$ matrix.  Suppose $\text{rank}(\mathbf{\Psi}) < J$ (e.g., the number of comparison groups could be less than the complexit of $\E[Y_{t^*_1}(0) \mid Y_{t^*_1-1}(0), G=g]$ or the comparison groups do not contain enough independent information relative to the complexist of $\E[Y_{t^*_1}(0) \mid Y_{t^*_1-1}(0), G=g]$), so that $\mathbf{\Psi}$ does not have full rank.  Notice that $t \in \mathcal{T}_{\text{mean}}$ implies that $\mathbf{\Psi} \beta = \begin{bmatrix} \E[Y_{t^*_1-1}(0) \mid G=g] \end{bmatrix}_{g \in \bar{\mathcal{G}}}$.  However, because $\mathbf{\Psi}$ does not have full rank, there are many solutions to this system of equations, and, in particular, there exists a solution $\beta^*$ such that $\E[Y_{t^*_1}(0) \mid G=1] = \E\big[ \psi(Y_{t^*_1-1}(0)) \mid G=1\big]'\beta \neq \E[Y_{t^*_1-1}(0) \mid G=1]$, implying $\tau^{\text{TH}}_t \neq \ATT_t$.

    Consider the following specific counterexample.  For some constants $\gamma$ and $k$, suppose that $Y_t(0) = Y_{t^*_1-1}(0) + \gamma\big(Y_{t^*_1-1}(0)^2 - k \big) + e$ where $e \mid Y_{t^*_1-1}(0),G=g \sim \N(0,\sigma_e^2)$.  Notice that this implies that $Y_t(0) \independent G \mid Y_{t^*_1-1}(0)$, so that lagged outcome unconfoundedness holds.  Further, suppose that $Y_{t^*_1-1}(0) \mid G=1 \sim \N(\mu_1,\sigma^2_1)$, and for $g \in \bar{\mathcal{G}}$, $Y_{t^*_1-1}(0) \mid G=g \sim \N(\mu_g, \sigma_g^2)$ and satisfying $\mu_g^2 + \sigma_g^2 = k$, but $\mu_1^2 + \sigma_g^2 \neq k$.  In this case,
    \begin{align*}
        \mathbf{\Psi} = \begin{bmatrix}
            1 & \mu_2 & \mu_2^2 + \sigma_2^2 \\
            1 & \mu_3 & \mu_3^2 + \sigma_3^2 \\
            \vdots & \vdots & \vdots \\
            1 & \mu_{|\mathcal{G}|} & \mu_{|\mathcal{G}|}^2 + \sigma_{|\mathcal{G}|}^2
        \end{bmatrix} = \begin{bmatrix}
            1 & \mu_2 & k \\
            1 & \mu_3 & k \\
            \vdots & \vdots & \vdots \\
            1 & \mu_{|\mathcal{G}|} & k
    \end{bmatrix}
    \end{align*}
    which has rank at most $2 < J=3$.  And, more specifically, it holds that
    \begin{align*}
        \tau^{\text{TH}}_t &= \E[Y_t(1) \mid G=1] - \E[Y_{t^*_1-1}(0) \mid G=1] \\
        &= \ATT_t + \E[Y_t(0) \mid G=1] - \E[Y_{t^*_1-1}(0) \mid G=1] \\
        &= \ATT_t + \E\big[ \E[Y_t(0) \mid Y_{t^*_1-1}(0), G=1] \bigm| G=1\big] - \E[Y_{t^*_1-1}(0) \mid G=1] \\
        &= \ATT_t + \E\big[ \E[Y_t(0) \mid Y_{t^*_1-1}(0), G=g] \bigm| G=1\big] - \E[Y_{t^*_1-1}(0) \mid G=1] \\
        &= \ATT_t + \E\Big[ Y_{t^*_1-1}(0) + \gamma\big(Y_{t^*_1-1}(0)^2 - k \big) \Bigm| G=1\Big] - \E[Y_{t^*_1-1}(0) \mid G=1] \\
        &= \ATT_t + \gamma(\mu_1^2 + \sigma_1^2 -k) \\
        &\neq \ATT_t
    \end{align*}
    where the first equality holds by the definition of $\tau^{\text{TH}}_t$, the second equality holds by adding and subtracting $\E[Y_t(0) \mid G=1]$, the third equality holds by the law of iterated expectations, the fourth equality holds by the lagged outcome unconfoundedness identification strategy, the fifth equality holds by the setting of the counterexample, and the last two equalities are immediate simplifications of the previous line.

    \bigskip

    \noindent \underline{\textit{(iv) Interactive Fixed Effects}:}
    Define
    \begin{align} \label{eqn:bar-lambda-ife-time-homogeneity}
        \bar{\mathbf{\Lambda}} = \begin{bmatrix} 1 & \E[\lambda'|G=g] \end{bmatrix}_{g \in \bar{\mathcal{G}}}
    \end{align}
    where we use the notation $\begin{bmatrix} & \cdot &  \end{bmatrix}_{g \in \bar{\mathcal{G}}}$ to indicate a matrix with one row for each $g \in \bar{\mathcal{G}}$; thus, $\bar{\mathbf{\Lambda}}$ is a $|\bar{\mathcal{G}}| \times (R+1)$ dimensional matrix, where $R$ continues to be defined as the number of factors in \Cref{eqn:ife}.  We assume that $\text{rank}(\bar{\mathbf{\Lambda}})=R+1$, which amounts to assuming that the groups are ``different enough'' from each other in terms of the mean of $\lambda$.  This also implies that the number of groups is larger than the number of interactive fixed effects.  That $t^*_1 \in \mathcal{T}_{\text{mean}}(1)$ implies that
    \begin{align*}
        \mathbf{0}_{|\bar{\mathcal{G}}|} &= \begin{bmatrix} \E[\Delta Y_{t^*_1}|G=g] \end{bmatrix}_{g \in \bar{\mathcal{G}}} \\
        &= \bar{\mathbf{\Lambda}} \begin{bmatrix} \Delta \theta_{t^*_1} \\ \Delta F_{t^*_1} \end{bmatrix}
    \end{align*}
    Since $\bar{\mathbf{\Lambda}}$ has full column rank, this implies that
    \begin{align} \label{eqn:ife-time-homogeneity-implication}
        \mathbf{0}_{R+1} = \begin{bmatrix} \Delta \theta_{t^*_1} \\ \Delta F_{t^*_1} \end{bmatrix}
    \end{align}
    Thus,
    \begin{align*}
        \ATT_{t^*_1} &= \E[Y_{t^*_1}(1) - Y_{t^*_1}(0) \mid G=1] \\
        &= \tau^{\text{TH}} - \E[Y_{t^*_1}(0) - Y_{t^*_1-1}(0) \mid G=1] \\
        &= \tau^{\text{TH}} + \Big( \Delta \theta_{t^*_1} + \E[\lambda' \mid G=1] \Delta F_{t^*_1} \Big) \\
        &= \tau^{\text{TH}}
    \end{align*}
    where the first equality holds by the definition of $\ATT_{t^*_1}$, the second equality holds by adding and subtracting $\E[Y_{t^*_1-1}(0) \mid G=1]$, the third equality holds from the interactive fixed effects model for untreated potential outcomes in \Cref{eqn:ife} and by \Cref{ass:ife-model}, and the last equality holds by \Cref{eqn:ife-time-homogeneity-implication}.

    \bigskip

    \noindent \underline{\textit{(v) Latent Unconfoundedness:}}
    Consider the following counterexample.  Suppose that $h_t(\eta_i) = \theta_t + \eta_i + \lambda_i'F_t + e_{it}$ where $R+1 > |\bar{\mathcal{G}}|$ (i.e., the number of latent factors plus one exceeds the number of comparison groups).  Then, mean stability for all comparison groups implies $\text{rank}(\bar{\mathbf{\Lambda}}) < R+1$ where $\bar{\mathbf{\Lambda}}$ is defined as in \Cref{eqn:bar-lambda-ife-time-homogeneity}.  Therefore, there exists a non-zero vector $\begin{bmatrix} \Delta \theta_t \\ \Delta F_t \end{bmatrix}'$ such that $\bar{\mathbf{\Lambda}} \begin{bmatrix} \Delta \theta_t \\ \Delta F_t \end{bmatrix} = \mathbf{0}_{|\bar{\mathcal{G}}|}$.  This implies that mean stability for all comparison groups does not force $\E[Y_{t^*_1}(0) - Y_{t^*_1-1}(0) \mid G=1] = 0$, and, hence, does not imply that $\tau^{\text{TH}}_{t^*_1} = \ATT_{t^*_1}$.

\end{proof}

\subsubsection{Additional Results for Change-in-Changes}

In line with the discussion in \Cref{sec:time-homogeneity}, define
\begin{align*}
    t_\text{dist}^{\text{TH}} &:=\max \big\{ s \in \{t^*_1, \ldots, \T\} \bigm| (Y_{s'}(0) | G=g) \overset{d}{=} (Y_{t^*_1-1}(0) | G=g) \ \text{for all } s'\leq s \text{ and } g \in \bar{\mathcal{G}}\big\}
\end{align*}
which is the last period after treatment adoption where the distribution of untreated potential outcomes is the same as in period $t^*_1-1$ for all groups in $\bar{\mathcal{G}}$.  Also, define
\begin{align*}
    \mathcal{T}_{\text{dist}} &= \big\{ t \in \{t^*_1, \ldots, \T\} \bigm| t \leq t_{\text{dist}}^{\text{TH}} \big\}
\end{align*}
which is the set of periods after treatment adoption where the distribution of untreated potential outcomes is the same as in period $t^*_1-1$ for all groups in $\bar{\mathcal{G}}$.  Then, we have the following result for $\tau^{\text{TH}}_t$ if the change-in-changes identification strategy and under ``distributional time homogeneity''.

\begin{proposition} \label{prop:cic-distributional-time-homogeneity}
    Suppose that the change-in-changes identification strategy holds.  Then, for any $t \in \mathcal{T}_{\text{dist}}$
    \begin{align*}
        \tau^{\text{TH}}_t = \ATT_t
    \end{align*}
\end{proposition}

\begin{proof}
    Notice that, for any group $g \in \bar{\mathcal{G}}$,
    \begin{align*}
        \tau^{\text{TH}}_t &= \E[Y_t(1) \mid G=1] - \E[Y_{t^*_1-1}(0) \mid G=1] \\
        &= \E[Y_t(1) \mid G=1] - \E\Big[Q_{Y_t(0)|G=g} \circ\F_{Y_{t^*_1-1}(0)|G=g}\big(Y_{t^*_1-1}(0)\big) \Bigm| G=1] \\
        &= \E[Y_t(1) \mid G=1] - \E\Big[Q_{Y_t(0)|G=1} \circ\F_{Y_{t^*_1-1}(0)|G=1}\big(Y_{t^*_1-1}(0)\big) \Bigm| G=1] \\
        &= \E[Y_t(1) \mid G=1] - \E\Big[Q_{Y_t(0)|G=1}(U) \Bigm| G=1] \\
        &= \E[Y_t(1) \mid G=1] - \E[Y_t(0) \mid G=1] \\
        &= \ATT_t
    \end{align*}
    where the first equality holds by the definition of $\tau^{\text{TH}}_t$, the second equality holds because $t \in \mathcal{T}_{\text{dist}}$ (i.e., the distribution of untreated potential outcomes is the same in period $t$ and period $t^*_1-1$ for group $g$), the third equality holds by the CiC identification strategy (\Cref{ass:cic} in particular), the fourth equality holds because $\F_{Y_{t^*_1-1}(0)|G=1}(Y_{t^*_1-1}(0))|G=1$ follows a $\text{Uniform}[0,1]$ distribution, the fifth equality holds by the probability inverse transform with $U \mid G=1 \sim \text{Uniform}[0,1]$, and the last equality holds by the definition of $\ATT_t$.
\end{proof}

\subsubsection{Additional Results for Lagged Outcome Unconfoundedness}

The next result shows that under lagged outcome unconfoundedness and with a flexible parametric model for $\E[Y_{t^*_1}(0) \mid Y_{t^*_1-1}(0), G=g]$, if time there are enough comparison groups that exhibit enough independent variation, then time homogeneity among the comparison groups implies that time homogeneity also holds for the treated group.

\begin{proposition} \label{prop:lou-mean-time-homogeneity} Suppose that $\E[Y_{t^*_1}(0) \mid Y_{t^*_1-1}(0), G=g] = \psi\big(Y_{t^*_1-1}(0)\big)'\beta$ where $\psi(y)$ is a $J$-dimensional vector of basis functions that includes the identity function, i.e., $\psi_j(y) = y$ for some $j$ (this requirement is met by typical polynomial, spline, or other flexible basis functions). Further, let $\mathbf{\Psi}$ be the $|\bar{\mathcal{G}}| \times J$matrix with rows $\E[\psi(Y_{t^*_1-1}(0)) \mid G=g]'$ for $g \in \bar{\mathcal{G}}$, and suppose that $\mathbf{\Psi}$ has full column rank $J$.  If $\text{rank}(\mathbf{\Psi}) = J$, then time homogeneity for the comparison groups implies that time mean time homogeneity also applies for the treated group, i.e., $\E[Y_{t^*_1}(0) \mid G=1] = \E[Y_{t^*_1-1}(0) \mid G=1]$, and, hence, $\tau^{\text{TH}}_{t^*_1} = \ATT_{t^*_1}$.
\end{proposition}

\begin{proof}
    Without loss of generality, assume that the first element of $\psi(y)$ is the identity function, i.e., $\psi_1(y) = y$.  Next, notice that time homogeneity implies that
    \begin{align*}
        \mathbf{\Psi} \beta &= \begin{bmatrix} \E[Y_{t^*_1-1}(0) \mid G=g] \end{bmatrix}_{g \in \bar{\mathcal{G}}}
    \end{align*}
    Since $\mathbf{\Psi}$ has full column rank, this system has a unique solution for $\beta$.  Moreover, notice that
    \begin{align*}
        \mathbf{\Psi} \begin{pmatrix} 1 \\ 0 \\ \vdots \\ 0 \end{pmatrix} &= \begin{bmatrix} \E[Y_{t^*_1-1}(0) \mid G=g] & \E[\psi_2(Y_{t^*_1-1}(0)) \mid G=g] & \cdots & \E[\psi_J(Y_{t^*_1-1}(0)) \mid G=g] \end{bmatrix}_{g \in \bar{\mathcal{G}}} \begin{pmatrix} 1 \\ 0 \\ \vdots \\ 0 \end{pmatrix} \\
        &= \begin{bmatrix} \E[Y_{t^*_1-1}(0) \mid G=g] \end{bmatrix}_{g \in \bar{\mathcal{G}}}
    \end{align*}
    This implies that the unique solution to the system is $\beta = \begin{pmatrix} 1 & 0 & \cdots & 0 \end{pmatrix}'$, and, hence,
    \begin{align*}
        \E[Y_{t^*_1}(0) \mid Y_{t^*_1-1}(0), G=g] &= \psi\big(Y_{t^*_1-1}(0)\big)'\beta = Y_{t^*_1-1}(0)
    \end{align*}
    for all $g \in \bar{\mathcal{G}}$.  Therefore,
    \begin{align*}
        \E[Y_{t^*_1}(0) \mid G=1] &= \E\big[ \E[Y_{t^*_1}(0) \mid Y_{t^*_1-1}(0), G=1] \mid G=1 \big] \\
        &= \E\big[ Y_{t^*_1-1}(0) \mid G=1 \big]
    \end{align*}
    which implies that $\tau^{\text{TH}}_{t^*_1} = \ATT_{t^*_1}$.
\end{proof}

The takeaway is that (i) in general, mean stability for multiple comparison groups does not guarantee mean stability for the treated group under lagged outcome unconfoundedness, but (ii) under a flexible parametric model for the transition function $\E[Y_{t^*_1}(0)\mid Y_{t^*_1-1}(0)]$ can restore mean stability for the treated group, and more comparison groups allow for more flexibility in $\E[Y_{t^*_1}(0)\mid Y_{t^*_1-1}(0)]$ while still delivering mean stability for the treated group.

\subsection{Estimation and Inference Proofs}

\begin{lemma}[Distance convergence rate] \label{lem:dim-error-rate}
    Under \Cref{ass:dist-if},
    \begin{align*}
        \max_{g\in\mathcal{G}} \big|\widehat d_g - d_g\big| = O_p\!\big(n_{\mathrm{min}}^{-1/2}\big).
    \end{align*}
    where $n_{\mathrm{min}}:=\min_{g\in\mathcal{G}} n_g$.
\end{lemma}
\begin{proof}[\textbf{Proof of \Cref{lem:dim-error-rate}}]
    This result is an immediate implication of \Cref{ass:dist-if} because $\mathcal{G}$ is a finite set.
\end{proof}

\bigskip

\begin{lemma}[No false inclusions] \label{lem:no-false-inclusions} Under \Cref{ass:sampling,ass:groups,ass:no-anticipation,ass:single-treated,ass:est-regularity,ass:est-group-separation,ass:est-bandwidth,ass:est-kernel,ass:dist-if}, if $g \notin \mathcal{G}^*$, then $g \notin \widehat{\mathcal{G}^*}$ as $n_{\mathrm{min}} \rightarrow \infty$. That is, for all $g \notin \mathcal{G}^*$, $\P\big(\K(\widehat{\d}_g/h_n) =0\big) \rightarrow 1$ as $n_{\mathrm{min}} \rightarrow \infty$, which implies $\K\big(\widehat{\d}_g/h_n\big) \overset{p}{\to} 0$ as $n_{\mathrm{min}} \rightarrow \infty$.
\end{lemma}
\begin{proof}[\textbf{Proof of Lemma \ref{lem:no-false-inclusions}}]

    Suppose that $g \notin \mathcal{G}^*$. By \Cref{ass:est-group-separation}, $\d_g \geq \d_{\mathrm{min}} > 0$.     Take any $\epsilon \in (0,\d_{\mathrm{min}})$. %
    Notice that:
    \begin{align}
        \P(\widehat{\d}_g \geq \d_{\mathrm{min}} - \epsilon) &= 1 - \P(\widehat{\d}_g < \d_{\mathrm{min}} - \epsilon) \notag \\
        &\geq 1 - \P(\widehat{\d}_g < \d_g - \epsilon) \notag \\
        &= 1 - \P( \big| \widehat{\d}_g - \d_g \big| > \epsilon) \notag \\
        & \to 1, \; \mathrm{as} \; n_{\mathrm{min}} \rightarrow \infty \label{eqn:SA.2.1}
    \end{align}
    where the first line holds by taking the complement; the second line holds since $\d_g \geq \d_{\mathrm{min}}$; the third line holds by rearranging terms and taking the absolute value; and the last line holds by \Cref{ass:dist-if}, which implies that $\big| \widehat{\d}_g - \d_g \big| = o_p(1)$.

    Next, take the ratio $\widehat{\d}_g / h_n$. On the event $\{\widehat{\d}_g \geq \d_{\mathrm{min}} - \epsilon\}$, which occurs with probability approaching one as $n_{\min} \to \infty$ by \Cref{eqn:SA.2.1}, we have that
    \begin{align}
        \frac{\widehat{\d}_g}{h_n} \geq \frac{\d_{\mathrm{min}} - \epsilon}{h_n} > 1 \label{eqn:lem2a}
    \end{align}
    where the second inequality holds for sufficiently large $n$ since $h_n \to 0$ and $\d_{\mathrm{min}}$ and $\epsilon$ are fixed with $\d_{\mathrm{min}} > \epsilon > 0$.  Next, notice that %

    \begin{align}
        \P\Big(\K\big(\widehat{\d}_g/h_n\big) = 0\Big) &\geq \P\big( \widehat{\d}_g/h_n > 1 \big) \notag \\
        &\geq \P\big( \widehat{\d}_g/h_n \geq (\d_{\mathrm{min}} - \epsilon)/h_n \big) \notag \\
        &= \P\big( \widehat{\d}_g \geq \d_{\mathrm{min}} - \epsilon \big) \notag \\
        & \to 1, \; \mathrm{as} \; n_{\mathrm{min}} \rightarrow \infty \label{eqn:SA.2.3}
    \end{align}
    where the first line holds since $u > 1 \implies \K(u)=0$ and because $\widehat{\d}_g/h_n > 1$ by construction, the second line holds by \Cref{eqn:lem2a} for $n$ large enough, the third line holds immediately by canceling the $h_n$ terms, and the fourth line holds by \Cref{eqn:SA.2.1} above. Thus, for any $g \notin \mathcal{G}^*$, $\P\big(\K(\widehat{d}_g/h_n\big) =0) \rightarrow 1$ as $n_{\mathrm{min}} \rightarrow \infty$.

\end{proof}

\begin{lemma}[Exclusion of non-close comparison groups from $\widehat{\mathcal{G}}^*$] \label{lem:exclusion-of-non-close-groups} Under \Cref{ass:sampling,ass:groups,ass:no-anticipation,ass:single-treated,ass:est-regularity,ass:est-group-separation,ass:est-bandwidth,ass:est-kernel,ass:dist-if}, the probability that any $g \notin \mathcal{G}^*$ is incorrectly included in the estimated set of close comparison groups, $\widehat{\mathcal{G}}^*$, converges to zero as $n_{\mathrm{min}} \rightarrow \infty$. That is, $\P\big( \widehat{\mathcal{G}}^* \subseteq \mathcal{G}^* \big) \to 1$ as $n_{\mathrm{min}} \rightarrow \infty$.
\end{lemma}
\begin{proof}[\textbf{Proof of \Cref{lem:exclusion-of-non-close-groups}}]
    Consider the probability that there exists at least some group $g \notin \mathcal{G}^*$ that receives positive weight:
    \begin{align}
        \P \left( \bigcup_{g \notin \mathcal{G}^*} \K( \widehat{\d}_g/h_n) > 0 \right) &\leq \sum_{g \notin \mathcal{G}^*} \P\big(\K(\widehat{\d}_g/h_n) > 0\big) \notag \\
        &= \sum_{g \notin \mathcal{G}^*} \Big(1 - \P\big(\K(\widehat{\d}_g/h_n) = 0\big)\Big) \notag \\
        & \to \sum_{g \notin \mathcal{G}^*} (1 - 1) \; \mathrm{as} \; n_{\mathrm{min}} \rightarrow \infty \notag \\
        &= 0 \label{eqn:SA.2.4}
    \end{align}
    where the first line holds by Boole's inequality; the second line holds by taking the complement and that $\K(\cdot)$ is non-negative; the third line holds by \Cref{eqn:SA.2.3}; and the fourth line holds trivially. That is, for all groups $g \notin \mathcal{G}^*$, the probability that non-credible groups receive zero weight converges to 1:
    \begin{align*}
        \P \left( \bigcap_{g \notin \mathcal{G}^*} \K(\widehat{\d}_g/h_n) = 0 \right) &= 1 - \P \left( \bigcup_{g \notin \mathcal{G}^*} \K(\widehat{\d}_g/h_n) > 0 \right) \to 1 \; \mathrm{as} \; n_{\mathrm{min}} \rightarrow \infty
    \end{align*}
    where the equality holds by De Morgan's law, and the convergence holds by \Cref{eqn:SA.2.4}. Hence, $ \P\big( \{ g \in \mathcal{G}^* : \K(\widehat{\d}_g / h_n) >0 \} \subseteq \mathcal{G}^*\big) = \P(\widehat{\mathcal{G}}^* \subseteq \mathcal{G}^*) \to 1$, as $n \to \infty$.%
    \end{proof}

\bigskip

\bigskip

\begin{lemma}[No false exclusions] \label{lem:no-false-exclusions}
Under \Cref{ass:sampling,ass:groups,ass:no-anticipation,ass:single-treated,ass:est-regularity,ass:est-group-separation,ass:est-bandwidth,ass:est-kernel,ass:dist-if},
for all $g \in \mathcal{G}^*$, $\K(\widehat{\d}_g/h_n) \overset{p}{\to} \K(0)$ as $n_{\mathrm{min}} \rightarrow \infty$.
Moreover, $\P\big(\K(\widehat{\d}_g/h_n) > 0\big) \to 1$ as $n_{\mathrm{min}} \to \infty$.
\end{lemma}

\begin{proof}[\textbf{Proof of Lemma \ref{lem:no-false-exclusions}}]
For any $g \in \mathcal{G}^*$, by definition, $d_g = 0$. From \Cref{ass:dist-if},
this implies $\widehat{d}_g = O_p(n_{\mathrm{min}}^{-1/2})$. For bandwidth $h_n > 0$:
\begin{align*}
    \frac{\widehat{d}_g}{h_n} = \frac{O_p(n_{\mathrm{min}}^{-1/2})}{h_n}
    = O_p\left(\frac{1}{\sqrt{n_{\mathrm{min}}} h_n}\right) \overset{p}{\to} 0
\end{align*}
where convergence in probability follows from \Cref{ass:est-bandwidth}, which ensures
$\sqrt{n_{\mathrm{min}}} h_n \to \infty$.

Since $\K$ is continuous at 0 by \Cref{ass:est-kernel}, the continuous mapping theorem implies:
\begin{align*}
    \K\big(\widehat{d}_g/h_n\big) \overset{p}{\to} \K(0)
\end{align*}
which completes the first part of the proof.  Next, for any $\epsilon \in \big(0, \K(0)\big)$, where $\K(0) > 0$ by \Cref{ass:est-kernel}, it follows that
\begin{align*}
    \P\Big(\K(\widehat{d}_g/h_n) > 0\Big) \geq \P\big(\K(\widehat{d}_g/h_n) > \epsilon\big) \to 1 \text{ as } n_{\mathrm{min}} \to \infty.
\end{align*}
\end{proof}

\bigskip

\begin{lemma}[Inclusion of all close comparison groups in $\widehat{\mathcal{G}}^*$] \label{lem:inclusion-of-all-close-groups} Under \Cref{ass:sampling,ass:groups,ass:no-anticipation,ass:single-treated,ass:est-regularity,ass:est-group-separation,ass:est-bandwidth,ass:est-kernel,ass:dist-if},
the probability that any group $g \in \mathcal{G}^*$ receives positive weight converges to one as $n_{\text{min}} \rightarrow \infty$.  Thus, $\P\big( \widehat{\mathcal{G}}^* \supseteq \mathcal{G}^* \big) \to 1$ as $n_{\mathrm{min}} \rightarrow \infty$.
\end{lemma}

\begin{proof}[\textbf{Proof of \Cref{lem:inclusion-of-all-close-groups}}]
    Consider the probability that all groups in $\mathcal{G}^*$ get positive weight
    \begin{align*}
        \P \left( \bigcap_{g \in \mathcal{G}^*} \K( \widehat{\d}_g/h_n) > 0 \right) &= 1 - \P \left( \bigcup_{g \in \mathcal{G}^*} \K( \widehat{\d}_g/h_n) = 0 \right) \\
        &\geq 1 - \sum_{g \in \mathcal{G}^*} \P\big(\K(\widehat{\d}_g/h_n\big) = 0) \\
        & \to 1 - \sum_{g \in \mathcal{G}^*} 0 \; \mathrm{as} \; n_{\mathrm{min}} \rightarrow \infty \\
        &= 1
    \end{align*}
    where the first line holds by De Morgan's law---taking the complement and since the kernel function is non-negative; %
    the second line holds by Boole's inequality; and the third line holds by \Cref{lem:no-false-exclusions}. This result states that the probability that all groups that belong in $\mathcal{G}^*$ are included in the estimated $\widehat{\mathcal{G}}^*$ converges to one as $n_{\mathrm{min}} \to \infty$, i.e., that $\P(\widehat{\C} \supseteq \C) \to 1$, as $n_{\text{min}} \to \infty$.
\end{proof}

\bigskip

\bigskip

\begin{proof}[\textbf{Proof of \Cref{prop:model-selection-consistency}}]
    The result holds immediately by \Cref{lem:exclusion-of-non-close-groups,lem:inclusion-of-all-close-groups}.
\end{proof}

\bigskip

\begin{proposition}\label{prop:consistent-estimator} Under \Cref{ass:sampling,ass:groups,ass:no-anticipation,ass:single-treated,ass:est-regularity,ass:est-group-separation,ass:est-bandwidth,ass:est-kernel,ass:dist-if}, our estimator of $\ATT_{t^*_1}$ is consistent, i.e.,
    \begin{align*}
        \widehat{\ATT}_{t^*_1}(h_n) \xrightarrow{p} \ATT_{t^*_1}
    \end{align*}
\end{proposition}
\begin{proof}[\textbf{Proof of Proposition \ref{prop:consistent-estimator}}]
    Notice that, for each $g \in \mathcal{G}$,
    \begin{align}
        \widehat{p}_g \cdot \K\big(\widehat{d}_g/h_n\big) \overset{p}{\to}
            \begin{cases}
                p_g \cdot \K(0), & g \in \mathcal{G}^*, \\
                0, & g \notin \mathcal{G}^*,
            \end{cases} \label{eqn:SA.4.1}
    \end{align}
    and further that
    \begin{align}
        \widehat{p}_g \cdot \K\big(\widehat{d}_g/h_n\big) \cdot \bar{Y}_{g,t^*_1} \overset{p}{\to}
            \begin{cases}
                p_g \cdot \K(0) \cdot \E[Y_{t^*_1} \mid G=g], & g \in \mathcal{G}^*, \\
                0, & g \notin \mathcal{G}^*,
            \end{cases} \label{eqn:SA.4.2}
    \end{align}
    which holds by \Cref{lem:no-false-inclusions,lem:no-false-exclusions} and Slutsky's theorem.  This implies that
    \begin{align*}
        \widehat{\ATT}_{t^*_1}(h_n) &= \bar{Y}_{1,t^*_1} - \sum_{g \in \mathcal{G}} \frac{ \widehat{p}_g \cdot \K(\widehat{\d}_g/h_n) }{ \sum_{g' \in \mathcal{G}} \widehat{p}_{g'} \cdot \K(\widehat{\d}_{g'}/h_n) } \cdot \bar{Y}_{g,t^*_1} \\
        &\overset{p}{\to} \E[Y_{t^*_1} \mid G=1] - \sum_{g \in \mathcal{G}^*} \left( \frac{p_g \cdot K(0)}{ \sum_{g' \in \mathcal{G}^*} p_{g'} \cdot K(0)} \right) \cdot \E[Y_{t^*_1} \mid G=g] \\
        &= \E[Y_{t^*_1} \mid G=1] - \sum_{g \in \mathcal{G}^*} \left( \frac{p_g}{ \sum_{g' \in \mathcal{G}^*} p_{g'}} \right) \cdot \E[Y_{t^*_1} \mid G=g] \\
        &= \ATT_{t^*_1}
    \end{align*}
\end{proof}

\bigskip

\begin{lemma}[Kernel convergence] \label{lem:kernel-convergence}
    Under \Cref{ass:sampling,ass:groups,ass:no-anticipation,ass:single-treated,ass:est-regularity,ass:est-group-separation,ass:est-bandwidth,ass:est-kernel,ass:dist-if},
    \begin{align*}
    \sqrt{n_{\mathrm{min}}} \sum_{g \in \mathcal{G}} \widehat{p}_g \cdot \left( \K(\widehat{d}_g / h_n) - \K(0) \indicator{g \in \G} \right) \cdot \bar{Y}_{g,t^*_1} \xrightarrow{p} 0.
    \end{align*}
\end{lemma}
\begin{proof}[\textbf{Proof of Lemma \ref{lem:kernel-convergence}}]
        Notice that we can decompose the sum into two parts:
        \begin{align*}
        & \sum_{g \in \mathcal{G}} \widehat{p}_g \cdot \left( \K(\widehat{\d}_g / h_n) - \K(0) \indicator{g \in \mathcal{G}^*} \right) \cdot \bar{Y}_{g,t^*_1} \\ 
        & \hspace{10mm} = \underbrace{ \sum_{g \in \G} \widehat{p}_g \cdot \left( \K(\widehat{\d}_g / h_n) - \K(0) \right) \cdot \bar{Y}_{g,t^*_1} }_{\mathrm{I}_n} + \underbrace{ \sum_{g \notin \G} \widehat{p}_g \cdot \left( \K(\widehat{\d}_g / h_n) - 0 \right) \cdot \bar{Y}_{g,t^*_1} }_{\mathrm{II}_n} 
        \end{align*}
        and we continue by showing both terms are $o_p(1)$.

        For term $\mathrm{I}_n$, by \Cref{ass:est-kernel}, $\K$ is a symmetric function that is twice continuously differentiable at 0, which implies that $\K'(0)=0$. By a Taylor expansion for each $g \in \mathcal{G}^*$, we have:
        \begin{align*}
            \K\left( \widehat{\d}_g / h_n \right) &= \K(0) + \K'(0) \cdot \left(\frac{ \widehat{\d}_g }{h_n} - 0 \right) + \frac{1}{2} \K''(0) \cdot \left( \frac{ \widehat{\d}_g }{h_n} - 0 \right)^2 + o_p \left( \left( \frac{ \widehat{\d}_g }{h_n} - 0 \right)^2 \right)
        \end{align*}
        which implies
        \begin{align}
            \K\left( \widehat{\d}_g / h_n \right) - \K(0) &= \K'(0) \cdot \left( \widehat{\d}_g / h_n \right) + \frac{1}{2} \K''(0) \cdot \left( \widehat{\d}_g / h_n \right)^2 + o_p \left( (\widehat{\d}_g / h_n) ^2 \right) \notag \\
            &=  \frac{1}{2} \K''(0) \cdot \left( \widehat{\d}_g / h_n \right)^2 + o_p \left( (\widehat{\d}_g / h_n) ^2 \right) \notag \\
            &= O_p\left( (\widehat{\d}_g / h_n)^2 \right) + o_p \left( (\widehat{\d}_g / h_n) ^2 \right) \notag \\
            &= O_p\left( (\widehat{\d}_g / h_n)^2 \right) \label{eqn:SA.7.1}
        \end{align}
        where the first equality holds immediately, the second equality holds since $\K'(0)=0$, the third and fourth equalities hold by properties of $O_p(\cdot)$.

        Recall that, by \Cref{lem:dim-error-rate}, $\widehat{\d}_g = O_p(n_{\mathrm{min}}^{-1/2})$; and, by \Cref{lem:no-false-exclusions}, that $\widehat{\d}_g / h_n = O_p \left( \frac{1}{ \sqrt{n_{\mathrm{min}}} h_n} \right)$. Hence, by properties of $O_p(\cdot)$, $ (\widehat{\d}_g / h_n)^2 = O_p \left( \frac{1}{ n_{\mathrm{min}} h_n^2} \right)$. Since $(1/2) \cdot K''(0)$ is some constant, from \Cref{eqn:SA.7.1}, we have:
        \begin{align}
            \K\left( \widehat{\d}_g / h_n \right) - \K(0) &= O_p \left( \frac{1}{ n_{\mathrm{min}} h_n^2} \right) \label{eqn:SA.7.2}
        \end{align}
        Next, $|\mathcal{G}| <\infty$ and $|\G| \neq 0$ by \Cref{ass:est-regularity}. For each $g \in  \mathcal{G}$, $\widehat{p}_g = O_p(1)$ and $\Bar{Y}_{g,t^*_1} = O_p(1)$ by the WLLN. Then, this implies:
        \begin{align}
            \sum_{g \in \G} \widehat{p}_g \cdot \Bar{Y}_{g,t^*_1} \cdot \left( \K(\widehat{\d}_g / h_n) - \K(0) \right) &= \sum_{g \in \G} O_p(1) \cdot O_p(1) \cdot O_p \left( \frac{1}{ n_{\mathrm{min}} h_n^2} \right) \notag \\
            &= \sum_{g \in \G} O_p \left( \frac{1}{ n_{\mathrm{min}} h_n^2} \right) \notag \\
            &= O_p \left( \frac{1}{ n_{\mathrm{min}} h_n^2} \right) \label{eqn:SA.7.3}
        \end{align}
        where the first equality holds by \Cref{eqn:SA.7.2}; and the second and third equalities hold by properties of $O_p(\cdot)$. Lastly, by multiplication of \Cref{eqn:SA.7.3} by $\sqrt{n_{min}}$, we have:
        \begin{align*}
            \sqrt{n_{min}} \cdot \sum_{g \in \G} \widehat{p}_g \cdot \Bar{Y}_{g,t^*_1} \cdot \left( \K(\widehat{d}_g / h_n) - \K(0) \right) &= O_p \left( \frac{1}{ \sqrt{n_{\mathrm{min}}} h_n^2} \right)
        \end{align*}
        by properties of $O_p(\cdot)$. Given \Cref{ass:est-bandwidth} which states $\sqrt{n_{\mathrm{min}}} h_n^2 \to \infty$, it holds that
        \begin{align}
            \sqrt{n_{min}} \sum_{g \in \G} \widehat{p}_g \cdot \Bar{Y}_{g,t^*_1} \cdot \left( K(\widehat{d}_g / h_n) - K(0) \right) &= o_p(1) \label{eqn:SA.7.4}
        \end{align}
        as desired. 
        For term $\mathrm{II}_n$, we consider only groups $g \notin \G$. %
        \Cref{lem:no-false-inclusions} states that for each group $g \notin \G$, $\K(\widehat{\d}_g/h_n) = 0$ with probability tending to one. 
        Since the number of groups is finite, we have that the sum
        \begin{align*}
            \sum_{g \notin \G} \widehat{p}_g \cdot \Bar{Y}_{g,t^*_1} \cdot \K(\widehat{\d}_g / h_n) = 0 
        \end{align*}
        with probability tending to one. This implies: 
        \begin{align}
            \sqrt{n_{\mathrm{min}}} \sum_{g \notin \G} \widehat{p}_g \cdot \Bar{Y}_{g,t^*_1} \cdot \K(\widehat{\d}_g / h_n) = o_p(1) \label{eqn:SA.7.5}
        \end{align}
        as desired.

        Finally, the result follows from \Cref{eqn:SA.7.4} and \Cref{eqn:SA.7.5}:
        \begin{align*}
            \sqrt{n_{\mathrm{min}}} \sum_{g \in \mathcal{G}} \widehat{p}_g \cdot \Bar{Y}_{g,t^*_1} \cdot \left( \K(\widehat{\d}_g / h_n) - \K(0) \indicator{g \in \G} \right) &= o_p( 1 ).
        \end{align*}
\end{proof}

\bigskip

\begin{lemma}[Asymptotic equivalence to oracle estimator] \label{lem:asymptotic-equivalence-oracle}
    Under \Cref{ass:sampling,ass:groups,ass:no-anticipation,ass:single-treated,ass:est-regularity,ass:est-group-separation,ass:est-bandwidth,ass:est-kernel,ass:dist-if},
    \begin{align*}
        \sqrt{n_{\mathrm{min}}} \Big( \widehat{\ATT}_{t^*_1} - \oracle{\ATT}_{t^*_1} \Big) \xrightarrow{p} 0.
    \end{align*}
\end{lemma}
\begin{proof}[\textbf{Proof of Lemma \ref{lem:asymptotic-equivalence-oracle}}]
    Notice that
    \begin{align*}
        & \sqrt{n_{\mathrm{min}}} \Big( \widehat{\ATT}_{t^*_1}(h_n) - \oracle{\ATT}_{t^*_1}\Big) \\
        &\hspace{50pt}= - \sqrt{n_{\mathrm{min}}} \sum_{g \in \mathcal{G}} \frac{\widehat{p}_g \cdot \K\big(\widehat{\d}_g / h_n \big)}{ \sum_{g' \in \mathcal{G}}\widehat{p}_{g'} \cdot \K\big(\widehat{\d}_{g'} / h_n\big)} \cdot \bar{Y}_{g,t^*_1} + \sqrt{n_{\mathrm{min}}} \sum_{g \in \mathcal{G}^*} \frac{\widehat{p}_g}{ \sum_{g' \in \mathcal{G}^*}\widehat{p}_{g'} } \cdot \bar{Y}_{g,t^*_1} \\
        &\hspace{50pt} = -\underbrace{\sqrt{n_{\mathrm{min}}}\left( \sum_{g \in \mathcal{G}} \frac{\widehat{p}_g \cdot \Big(\K\big(\widehat{\d}_g / h_n \big) - \K(0) \indicator{g \in \mathcal{G}^*}\Big)}{ \sum_{g' \in \mathcal{G}}\widehat{p}_{g'} \cdot \K\big(\widehat{\d}_{g'} / h_n\big)} \cdot \bar{Y}_{g,t^*_1}\right)}_{\text{I}} \\
        & \hspace{62pt} + \underbrace{\sqrt{n_{\mathrm{min}}}\left( \sum_{g \in \mathcal{G}} \frac{\widehat{p}_g \cdot \K(0) \indicator{g \in \mathcal{G}^*} }{ \sum_{g' \in \mathcal{G}}\widehat{p}_{g'} \cdot \K(0) \indicator{g' \in \mathcal{G}^*}} \cdot \bar{Y}_{g,t^*_1} - \sum_{g \in \mathcal{G}} \frac{ \widehat{p}_g \cdot \K(0) \indicator{g \in \mathcal{G}^*} }{ \sum_{g' \in \mathcal{G}}\widehat{p}_{g'} \cdot \K\big(\widehat{\d}_{g'} / h_n\big)} \cdot \bar{Y}_{g,t^*_1} \right)}_{\text{II}}
    \end{align*}
    where the first equality holds by the definitions of $\widehat{\ATT}_{t^*_1}(h_n)$ and $\oracle{\ATT}_{t^*_1}$; and the second equality holds by adding and subtracting $\sqrt{n_{\mathrm{min}}} \displaystyle \sum_{g \in \mathcal{G}} \frac{\widehat{p}_g \cdot \K(0) \indicator{g \in \mathcal{G}^*} }{ \sum_{g' \in \mathcal{G}}\widehat{p}_{g'} \cdot \K\big(\widehat{\d}_{g'} / h_n\big)} \cdot \bar{Y}_{g,t^*_1}$.
    Term I converges to 0 by \Cref{lem:kernel-convergence} and the continuous mapping theorem. For Term II, notice that
    \begin{align*}
        \text{II} &= \sqrt{n_{\mathrm{min}}} \sum_{g \in \mathcal{G}} \frac{\widehat{p}_g \cdot \K(0) \cdot \indicator{g \in \mathcal{G}^*}}{\sum_{g' \in \mathcal{G}}\widehat{p}_{g'} \cdot \K(0) \indicator{g' \in \mathcal{G}^*}} \cdot \bar{Y}_{g,t^*_1} \cdot \left( 1 - \frac{\sum_{g' \in \mathcal{G}}\widehat{p}_{g'} \cdot \K(0) \indicator{g' \in \mathcal{G}^*}}{\sum_{g' \in \mathcal{G}}\widehat{p}_{g'} \cdot \K\big(\widehat{\d}_{g'} / h_n\big)} \right) \\
        &= \sqrt{n_{\mathrm{min}}} \, \oracle{\ATT}_{t^*_1} \cdot \left( \frac{\sum_{g \in \mathcal{G}}\widehat{p}_g \cdot \Big(\K\big(\widehat{\d}_g / h_n\big) - \K(0) \indicator{g \in \mathcal{G}^*}\Big)}{\sum_{g' \in \mathcal{G}}\widehat{p}_{g'} \cdot \K\big(\widehat{\d}_{g'} / h_n\big)} \right) \xrightarrow{p} 0
    \end{align*}
    where the first equality holds by rearranging terms; the second equality holds by the definition of $\oracle{\ATT}_{t^*_1}$; and the convergence to 0 holds by \Cref{lem:kernel-convergence}, the continuous mapping theorem, and the fact that the denominator converges to a positive constant, as shown in the proof of  \Cref{prop:consistent-estimator}.

\end{proof}

\begin{lemma}[Asymptotic normality of oracle estimator] \label{lem:asymptotic-normality-oracle}
    Under \Cref{ass:sampling,ass:groups,ass:no-anticipation,ass:single-treated,ass:est-regularity,ass:est-group-separation,ass:est-bandwidth,ass:est-kernel,ass:dist-if},
    \begin{align*}
        \sqrt{n}\left( \oracle{\ATT}_{t^*_1} - \ATT_{t^*_1}\right) &= \frac{1}{\sqrt{n}} \sum_{i=1}^n \psi(W_i) + o_p(1) \\
        & \xrightarrow{d} \N(0,V)
    \end{align*}
    where
    \begin{align*}
    \psi(W) &= \frac{\indicator{G = 1}}{p_1} \big(Y_{t^*_1} - \E[Y_{t^*_1} \mid G=1] \big) - \frac{\indicator{G \in \mathcal{G}^*}}{p_{\mathcal{G}^*}} \big(Y_{t^*_1} - \E[Y_{t^*_1} \mid G \in \mathcal{G}^*] \big) \\[4pt]
    \text{and } V &= \E\big[ \psi(W)^2 \big] = \frac{\Var(Y_{t^*_1} \mid G=1)}{p_1} + \frac{\Var(Y_{t^*_1} \mid G \in \mathcal{G}^*)}{p_{\mathcal{G}^*}}
    \end{align*}
\end{lemma}
\begin{proof}
    Recall that
    \begin{align*}
        \oracle{\ATT}_{t^*_1} &= \frac{1}{n_1} \sum_{i=1}^n \indicator{G_i = 1} Y_{it^*_1} - \frac{1}{n_{\mathcal{G}^*}} \sum_{i=1}^n \indicator{G_i \in \mathcal{G}^*} Y_{it^*_1}
    \end{align*}
    where $n_{\mathcal{G}^*}$ is the number of units in any group in $\mathcal{G}^*$, and this implies that
    \begin{align*}
        & \sqrt{n} \left( \oracle{\ATT}_{t^*_1} - \ATT_{t^*_1}\right) \\
        & \hspace{50pt}= \frac{n}{n_1} \cdot \frac{1}{\sqrt{n}} \sum_{i=1}^n \indicator{G_i = 1} \big(Y_{it^*_1} - \E[Y_{t^*_1} \mid G=1] \big) \\
        & \hspace{62pt} - \frac{n}{n_{\mathcal{G}^*}} \cdot \frac{1}{\sqrt{n}} \sum_{i=1}^n \indicator{G_i \in \mathcal{G}^*} \big(Y_{it^*_1} - \E[Y_{t^*_1} \mid G \in \mathcal{G}^*] \big) \\
        & \hspace{50pt}= \frac{1}{\sqrt{n}} \sum_{i=1}^n \left( \frac{\indicator{G_i = 1}}{p_1} \big(Y_{it^*_1} - \E[Y_{t^*_1} \mid G=1] \big) - \frac{\indicator{G_i \in \mathcal{G}^*}}{p_{\mathcal{G}^*}} \big(Y_{it^*_1} - \E[Y_{t^*_1} \mid G \in \mathcal{G}^*] \big) \right) + o_p(1) \\
    \end{align*}
    where the first equality holds by the definition of $\ATT_{t^*_1}$ and by multiplying and dividing by $\sqrt{n}$, and the second equality holds by the law of large numbers, continuous mapping theorem, and combining terms. The result then follows by applying the central limit theorem.
\end{proof}

\bigskip

\begin{proof}[\textbf{Proof of \Cref{thm:asymptotic-normality}}]
    The result follows immediately from \Cref{lem:asymptotic-equivalence-oracle} and \Cref{lem:asymptotic-normality-oracle}.
\end{proof}

\end{document}